\documentclass[aps, pra, reprint, floatfix, superscriptaddress]{revtex4-2}

\usepackage{amsmath, amssymb, amsfonts, amsthm, bm, bbm, graphicx, listings, textcomp, mathtools}
\usepackage[caption=false]{subfig}
\usepackage[colorlinks=true,linkcolor=blue,citecolor=blue,urlcolor=blue]{hyperref}

\newcommand{\github}[1]{\href{https://github.com/MichaelLiu2024/Entropic-Continuity-Bounds-and-Quantum-Distance-Measures}{#1}}

\newcommand{\bra}[1]{\left\langle #1 \right|}
\newcommand{\ket}[1]{\left| #1 \right\rangle}
\newcommand{\braket}[2]{\left\langle #1 \middle| #2 \right\rangle}
\newcommand{\ketbra}[2]{\ket{#1} \! \bra{#2}}

\newcommand{\set}[1]{\left\{ #1 \right\}}
\newcommand{\parenth}[1]{\left( #1 \right)}
\newcommand{\abs}[1]{\left| #1 \right|}
\newcommand{\norm}[1]{\left\lVert #1 \right\rVert}

\newcommand{\tr}[1]{\operatorname{tr} \parenth{#1}}
\newcommand{\trA}[1]{\operatorname{tr}_A \parenth{#1}}

\newcommand{\eps}{\varepsilon}

\newcommand{\T}[2]{\operatorname{T} \parenth{#1, #2}}
\newcommand{\F}[2]{\operatorname{F} \parenth{#1, #2}}
\newcommand{\Tc}[2]{\operatorname{T}_c \parenth{#1, #2}}
\newcommand{\Fc}[2]{\operatorname{F}_c \parenth{#1, #2}}
\newcommand{\A}[2]{\operatorname{A} \parenth{#1, #2}}
\newcommand{\G}[0]{\operatorname{G}}

\newcommand{\srss}[0]{\sqrt{\rho} \, \sqrt{\sigma}}
\newcommand{\power}[2]{#1^{#2\frac{1}{2}}}
\newcommand{\K}[0]{\sqrt{\sqrt{\rho} \, \sigma \sqrt{\rho}}}
\newcommand{\KU}[0]{\sqrt{\rho} \, \sqrt{\sigma} \, U}
\newcommand{\KxU}[0]{\sqrt{\rho} \, \ketbra{e_x}{e_x} \sqrt{\sigma} \, U}
\newcommand{\M}[0]{\power{\rho}{-} \K \, \power{\rho}{-}}
\newcommand{\Hvn}[1]{\operatorname{H} \parenth{#1}}
\newcommand{\hbinary}[1]{\operatorname{h} \parenth{#1}}
\newcommand{\HCond}[1]{\operatorname{H} \parenth{A|B}_{#1}}
\newcommand{\Hc}[0]{\operatorname{H}_c}
\newcommand{\spec}[0]{\operatorname{spec}}

\newcommand{\Span}[0]{\operatorname{Span}}
\newcommand{\lambertW}[1]{\operatorname{W}_0 \parenth{#1}}

\newcommand{\sumj}[0]{\sum_{j=1}^{d_A}}
\newcommand{\sumk}[0]{\sum_{k=1}^{d_B}}
\newcommand{\suml}[0]{\sum_{l=1}^{d_A}}
\newcommand{\sumkj}[0]{\sumk \sumj}

\newcommand{\xxx}[3]{#1(\theta)_{#2 #3}}

\newcommand{\xinX}[0]{x \in \mathcal{X}}
\newcommand{\yinX}[0]{y \in \mathcal{X}}

\newcommand{\LipschitzConstant}[1]{
    \begin{cases}
        \sqrt{8 \frac{\ln{x_0}}{x_0}} \sqrt{#1 - 1} & 1 \leq #1 \leq 4 \\
        2 \ln{#1} & #1 \geq 5
    \end{cases}
}

\newcommand{\negl}{\operatorname{negl}}

\newtheorem{theorem}{Theorem}
\newtheorem{lemma}{Lemma}
\newtheorem{condition}{Condition}

\begin{document}

\title{Lipschitz continuity of quantum-classical conditional entropies with respect to angular distance, and related properties of angular distance}

\author{Michael Liaofan Liu}
\email{mliu24@amherst.edu}
\affiliation{Institute for Quantum Computing, University of Waterloo, Waterloo, Ontario, Canada N2L 3G1}
\affiliation{Department of Mathematics, Amherst College, Amherst, MA 01002, USA}
\author{Florian Kanitschar}
\affiliation{Institute for Quantum Computing, University of Waterloo, Waterloo, Ontario, Canada N2L 3G1}
\affiliation{Technische Universität Wien, Faculty of Mathematics and Geoinformation, Wiedner Hauptstraße 8, 1040 Vienna, Austria}
\author{Amir Arqand}
\affiliation{Institute for Quantum Computing, University of Waterloo, Waterloo, Ontario, Canada N2L 3G1}
\author{Ernest Y.-Z. Tan}
\affiliation{Institute for Quantum Computing, University of Waterloo, Waterloo, Ontario, Canada N2L 3G1}

\date{\today}

\begin{abstract}
We derive a Lipschitz continuity bound for quantum-classical conditional entropies with respect to angular distance, with a Lipschitz constant that is independent of the dimension of the conditioning system. This bound is sharper in some situations than previous continuity bounds, which were either based on trace distance (where Lipschitz continuity is not possible), or based on angular distance but did not include a conditioning system. However, we find that the bound does not directly generalize to fully quantum conditional entropies. To investigate possible counterexamples in that setting, we study the characterization of states which saturate the Fuchs--van de Graaf inequality and thus have angular distance approximately equal to trace distance. We give an exact characterization of such states in the invertible case. For the noninvertible case, we show that the situation appears to be significantly more elaborate, and seems to be strongly connected to the question of characterizing the set of fidelity-preserving measurements.
\end{abstract}

\maketitle

\section{Introduction}
Given two quantum states $\rho$ and $\sigma$ on a Hilbert space $\mathcal{H}$, one of the most natural questions to ask is how similar $\rho$ and $\sigma$ are. Common measures to answer this question include the trace distance,
\begin{equation}\label{TraceDistanceDef}
    \T{\rho}{\sigma} \coloneqq \frac{1}{2} \norm{\rho - \sigma}_1,
\end{equation}
and the (root-)fidelity,
\begin{equation}\label{FidelityDef}
    \F{\rho}{\sigma} \coloneqq \norm{\srss}_1.
\end{equation}
The trace distance is a metric on the set of density operators $\mathcal{D}(\mathcal{H})$, and it has a meaningful interpretation as the distinguishability of two quantum states. In a quantum hypothesis testing scenario, where Bob randomly prepares one of two states $\rho$ and $\sigma$ (with equal probability) for Alice to distinguish, Alice can correctly identify the incoming state with probability $\frac{1 + \T{\rho}{\sigma}}{2}$. In contrast, the fidelity is not a metric, but it can be interpreted as the probability that a state $\rho$ ``passes a test" for being the same as a pure state $\sigma$ \cite{MW13}.

Another important task in quantum information theory is to quantify the amount of information present in a quantum system. The von Neumann entropy
\begin{equation*}
    \Hvn{\rho} \coloneqq - \tr{\rho \ln \rho}
\end{equation*}
is one quantity which fulfills this role \footnote{In this work, we define entropies via the natural logarithm rather than the base-$2$ logarithm for ease of presentation in the proofs.}, as it appears in many fundamental information theoretic tasks such as Schumacher data compression \cite{CD96} and randomness extraction \cite{BFW12}. This concept can be extended to conditional entropies $\HCond{\rho}$ for bipartite states $\rho \coloneqq \rho_{A B} \in \mathcal{D}(\mathcal{H}^A \otimes \mathcal{H}^B)$, with one of several equivalent definitions being the difference between the joint entropy and the marginal entropy,
\begin{equation}\label{ConditionalEntropyDefinition}
    \HCond{\rho} \coloneqq \Hvn{\rho_{AB}} - \Hvn{\rho_B},
\end{equation}
where $\rho_B \coloneqq \trA{\rho}$ is the reduced state of $\rho$ on $\mathcal{H}^B$. Further details about quantum distance measures and quantum entropies can be found in e.g.\ \cite{NC10, MW13}.

A useful property of the von Neumann entropy is that it is continuous for finite-dimensional quantum systems. This motivates the search for so-called entropic continuity bounds, which capture the notion that two states $\rho, \sigma$ close in some metric $\operatorname{d}$, e.g.\ $\operatorname{d}(\rho, \sigma) = \delta \gtrsim 0$, are expected to be close in entropy as well, i.e.\ 
\[
\abs{\Hvn{\rho} - \Hvn{\sigma}} \leq f(\delta),
\]
where $f$ is some function such that $\lim_{\delta \to 0} f(\delta) = 0$.

For example, in~\cite{KA07}, Audenaert derived the tightest form of the Fannes-type continuity bound for the von Neumann entropy in terms of trace distance. Specifically, letting $d \coloneqq \dim{\mathcal{H}} \in \mathbb{N}$, $\rho, \sigma \in \mathcal{D}(\mathcal{H})$, and $T \coloneqq \T{\rho}{\sigma}$, Audenaert showed that
\begin{equation}\label{AudenaertBound}
    \abs{\Hvn{\rho} - \Hvn{\sigma}} \leq T \ln\parenth{d-1} + \hbinary{T},
\end{equation}
where $\operatorname{h}(x) \coloneqq -x\ln x - (1-x)\ln (1-x)$ is the binary entropy function.

Similar continuity bounds also exist for conditional entropies. As shown by Winter~\cite{Win16}, letting $d_A \coloneqq \dim \mathcal{H}^A \in \mathbb{N}$, $d_B \coloneqq \dim \mathcal{H}^B \in \mathbb{N}$, $\rho, \sigma \in \mathcal{D}\parenth{\mathcal{H}^A \otimes \mathcal{H}^B}$, and $T \coloneqq \T{\rho}{\sigma}$, the following holds:
\begin{equation}\begin{aligned}\label{WinterBound}
    \Big| \! \HCond{\rho} &- \HCond{\sigma} \! \Big| \\
    &\leq 2 T \ln d_A + \parenth{1+T} \hbinary{\frac{T}{1+T}}.
\end{aligned}\end{equation}

Such continuity bounds have been applied in various contexts. For example, in~\cite{Upadhyaya_2021}, Upadhyaya et al.\ constructed a finite-dimensional cutoff formulation for a class of infinite-dimensional entropy optimization problems. Qualitatively, that work argues that if an infinite-dimensional state is ``close'' (under some metric) to a finite-dimensional state, then an entropic continuity bound allows us to replace the former with the latter and compensate for the resulting change in entropy by applying a correction based on the continuity bound.
This so-called dimension-reduction method plays an important role in quantum key distribution (QKD) security proofs~\cite{KGL+23}. However, it relies heavily on the continuity bound in Eq.~(\ref{WinterBound}) to compute the required correction term. An improved continuity bound would lead to a smaller correction term in this method and hence a larger secret key rate.

Another application of entropic continuity bounds arises in unstructured entropy optimization problems, as studied in e.g.~\cite{SBV+21}. In that work, the approach is that in order to minimize the entropy over some set of states, one simply computes the entropy on a sufficiently fine discrete ``grid'' of states in the set, then uses the continuity bound to ensure that the true minimum does not lie more than $f(\delta)$ away from the minimum over the grid. Again, an improved continuity bound would result in tighter results from such an approach.

In the above contexts, two desirable properties of the continuity bound $f(\delta)$ (for conditional entropies) are as follows.
\begin{condition}\label{DesiredBoundConditionIndependence}
$f(\delta)$ should be independent of $d_B$, the dimension of the conditioning system $\mathcal{H}^B$.
\end{condition}
\begin{condition}\label{DesiredBoundConditionDifferentiable}
$f(\delta)$ should have finite (and ideally small) derivative
at $\delta = 0$.
\end{condition}
\noindent The first property is useful (or in some cases required) for the applications mentioned above, since in those contexts the conditioning system may have large or unbounded dimension. The second property is desirable for obtaining better scaling at small $\delta$, since then we would not require extremely small values of $\delta$ in order to force the entropy difference to be small.

While the Winter bound (Eq.~(\ref{WinterBound})) satisfies condition~\ref{DesiredBoundConditionIndependence}, it does not satisfy condition~\ref{DesiredBoundConditionDifferentiable} due to the binary entropy term $\operatorname{h}$, which has unbounded derivative as $\delta \to 0$. In fact, such scaling of the conditional entropy with respect to trace distance is in some sense unavoidable, since there is an explicit family of states that saturates the Audenaert bound (Eq.~(\ref{AudenaertBound})), which has the binary entropy term as well. To work around this issue and obtain a bound that satisfies both conditions \ref{DesiredBoundConditionIndependence} and  \ref{DesiredBoundConditionDifferentiable}, one approach is to consider an alternative distance measure such as the angular distance, defined as
\begin{equation*}
    \A{\rho}{\sigma} \coloneqq \arccos{\F{\rho}{\sigma}}.
\end{equation*}
We remark that this is not simply an arbitrary change of distance measure: in the context of the applications mentioned above, the quantity that arises ``naturally'' in the analysis is the fidelity rather than the trace distance, hence working with the bound in Eq.~(\ref{WinterBound}) is somewhat suboptimal.

This approach is promising in light of the following result. In \cite{SBV+21}, Sekatski et al.\ proved Lipschitz continuity of the von Neumann entropy with respect to angular distance. That is, for $d \coloneqq \dim{\mathcal{H}} \in \mathbb{N}$, $\rho, \sigma \in \mathcal{D}(\mathcal{H})$, and $x_0 \coloneqq \exp{\parenth{\lambertW{-\frac{2}{e}}}} \approx 4.922$, where $\operatorname{W}_0$ is the principal branch of the Lambert-W function, it was shown that
\begin{equation}\label{SekatskiBound}
    \abs{\Hvn{\rho} - \Hvn{\sigma}} \leq u(d) \A{\rho}{\sigma},
\end{equation}
where the Lipschitz constant $u(d)$ is
\begin{equation}\label{LipschitzConstant}
    u(d) \coloneqq \LipschitzConstant{d}.
\end{equation}

Now, a naive application of Eq.~(\ref{SekatskiBound}) to conditional entropies, using Eq.~(\ref{ConditionalEntropyDefinition}) and the triangle inequality, would yield
\begin{equation}\begin{aligned}\label{SekatskiBoundConditional}
    \Big| \! \HCond{\rho} &- \HCond{\sigma} \! \Big| \\
    &\leq u(d_A d_B) \A{\rho}{\sigma} + u(d_B) \A{\rho_B}{\sigma_B} \\
    &\leq \parenth{u(d_A d_B) + u(d_B)} \A{\rho}{\sigma},
\end{aligned}\end{equation}
where in the last line we used the monotonicity of the angular distance under quantum channels. While this bound satisfies condition~\ref{DesiredBoundConditionDifferentiable}, it violates condition $\ref{DesiredBoundConditionIndependence}$. However, the estimates to obtain Eq.~(\ref{SekatskiBoundConditional}) from Eq.~(\ref{SekatskiBound}) are crude and leave room for refinement. Thus, we ask whether it is possible to obtain Lipschitz continuity of the conditional entropy with respect to angular distance, while avoiding dependence on $d_B$ in the final bound. 

In this work, we answer this question in the affirmative when $\rho$ and $\sigma$ are quantum-classical states on $\mathcal{H}^A \otimes \mathcal{H}^B$ (i.e.\ when there exists an orthonormal basis $\set{\ket{g_k}}_k$ for $\mathcal{H}^B$ such that both $\rho$ and $\sigma$ are of the form $\sum_k \gamma_k \tau_k \otimes \ketbra{g_k}{g_k}$ for some density operators $\tau_k \in \mathcal{D}(\mathcal{H}^A)$ and probabilities $\gamma_k \in [0,1]$). We present this result in Sec.~\ref{sec:contbnd}. However, we find that our bound does not hold in general for fully quantum states. To further investigate counterexamples in this setting, we study characterizations of states saturating the Fuchs--van de Graaf inequalities. In particular, the states saturating the upper bound in the inequality have $\T{\rho}{\sigma}\approx \A{\rho}{\sigma}$ when $\A{\rho}{\sigma}$ is small, so these states could pose an obstruction to deriving continuity bounds in terms of $\A{\rho}{\sigma}$ that scale better than those in terms of $\T{\rho}{\sigma}$. While it is well-known that any pair of pure states saturate the upper Fuchs--van de Graaf inequality, we show that these are not the only such states. In Sec.~\ref{sec:FvdG}, we provide a characterization of all such pairs $(\rho,\sigma)$ in the case where both of them are invertible. This result may be of independent interest in other applications such as computing QKD keyrates (we discuss this further in the appendices). However, we find that such a characterization in the general case where $(\rho,\sigma)$ are noninvertible appears significantly more challenging, and we discuss how it relates to identifying the set of measurements that preserve the fidelity between states. Finally, we provide some concluding remarks in Sec.~\ref{sec:conclusion}.

\section{Continuity Bound}\label{sec:contbnd}

We now state and prove the main result of our manuscript, a continuity bound for the conditional entropy of quantum-classical states with respect to angular distance. Subsequently, we discuss the tightness of this bound, and we highlight some challenges for generalizing our result to classical-quantum or fully quantum states.

\subsection{Main theorem and proof}

\begin{theorem}
Let $\mathcal{H}^A$ and $\mathcal{H}^B$ be Hilbert spaces of finite dimension $d_A$ and $d_B$, respectively. Let $\rho, \sigma \in \mathcal{D}\parenth{\mathcal{H}^A \otimes \mathcal{H}^B}$. Let $u(\cdot)$ and $x_0$ be defined as in Eq.~(\ref{LipschitzConstant}) and the preceding text. Suppose in addition that $\rho$ and $\sigma$ are both quantum-classical states with respect to $\mathcal{H}^A$ and $\mathcal{H}^B$. Then
\begin{equation}\label{ContinuityBound}
    \abs{\HCond{\rho} - \HCond{\sigma}} \leq u(d_A) \A{\rho}{\sigma}.
\end{equation}
\end{theorem}

\begin{proof}
Since $\rho$ and $\sigma$ are quantum-classical states, we can write
\begin{equation*}\begin{aligned}
    \rho &= \sumk \alpha_k \rho_k \otimes \ketbra{f_k}{f_k} \\
    \sigma &= \sumk \beta_k \sigma_k \otimes \ketbra{f_k}{f_k}
\end{aligned}\end{equation*}
for some density operators $\rho_k, \sigma_k \in \mathcal{D}(\mathcal{H}^A)$, probabilities $\alpha_k, \beta_k \in [0,1]$ which satisfy $\sumk \alpha_k = 1 = \sumk \beta_k$, and orthonormal basis $\set{\ket{f_k}}_k$ for $\mathcal{H}^B$. For each $k$, consider a spectral decomposition of $\rho_k$ and $\sigma_k$,
\begin{equation*}\begin{aligned}
    \rho_k &= \sumj p_{j k} \ketbra{e_{j k}}{e_{j k}} \\
    \sigma_k &= \sumj q_{j k} \ketbra{\tilde{e}_{j k}}{\tilde{e}_{j k}},
\end{aligned}\end{equation*}
where the eigenvalues $p_{j k},q_{j k} \geq 0$ satisfy $\sumj p_{j k} = 1 = \sumj q_{j k}$, and the eigenvectors form orthonormal bases $\set{\ket{e_{j k}}}_j, \set{\ket{\tilde{e}_{j k}}}_j$ for $\mathcal{H}^A$. Defining $\rho_{j k} \coloneqq p_{j k} \alpha_k$ and $\sigma_{j k} \coloneqq q_{j k} \beta_k$ for all $j$ and $k$, $\rho$ and $\sigma$ can be written as
\begin{equation*}\begin{aligned}
    \rho &= \sumkj \rho_{j k} \ketbra{e_{j k}}{e_{j k}} \otimes \ketbra{f_k}{f_k} \\
    \sigma &= \sumkj \sigma_{j k} \ketbra{\tilde{e}_{j k}}{\tilde{e}_{j k}} \otimes \ketbra{f_k}{f_k},
\end{aligned}\end{equation*}
and their partial traces can be written as
\begin{equation*}\begin{aligned}
    \rho_B &\coloneqq \trA{\rho} = \sumk \parenth{\sumj \rho_{j k}} \ketbra{f_k}{f_k} \\
    \sigma_B &\coloneqq \trA{\sigma} = \sumk \parenth{\sumj \sigma_{j k}} \ketbra{f_k}{f_k}.
\end{aligned}\end{equation*}

Now, observe that the eigenvalues $\rho_{j k}$ and $\sigma_{j k}$ of $\rho$ and $\sigma$ completely determine the eigenvalues of their partial traces $\rho_B$ and $\sigma_B$, respectively. This allows us to ``map" the problem to $\mathbb{R}^d$, where $d \coloneqq d_A d_B$, as follows. For each $k \in \set{1,\ldots,d_B}$, let us choose the ordering of the eigenvalues $p_{jk}$ (and corresponding eigenvectors $\ket{e_{j k}}$) to be such that $p_{1k} \geq p_{2k} \geq ... \geq p_{d_A k}$; similarly, choose the ordering of the eigenvalues $q_{jk}$ to be such that $q_{1k} \geq q_{2k} \geq ... \geq q_{d_A k}$. Now, consider the vectors
\begin{equation}\label{eq:rsdefn}
\begin{aligned}
    r &\coloneqq \parenth{\sqrt{\rho_{j k}}}_{k,j} \\
    s &\coloneqq \parenth{\sqrt{\sigma_{j k}}}_{k,j}
\end{aligned}
\end{equation}
in $\mathbb{R}^d$, where the entries of $r$ and $s$ are ordered with $k$ as the outer index and $j$ as the inner index. We observe that the angular distance $\A{\rho}{\sigma}$ between $\rho$ and $\sigma$ is always lower bounded by the angular distance $\theta_0 \coloneqq \arccos\parenth{r \cdot s} \in [0, \frac{\pi}{2}]$ between $r$ and $s$. To see this, we decompose the fidelity as a sum over $k$ using the quantum-classical structure, then apply a variational characterization of the trace norm \cite{MW13} and the von Neumann trace inequality \cite{Mirsky75}, which yields
\begin{align*}
    \norm{\srss}_1 &= \sumk \sqrt{\alpha_k \beta_k} \norm{\sqrt{\rho_k}\sqrt{\sigma_k}}_1\\ 
    &= \sumk \sqrt{\alpha_k \beta_k} \abs{\tr{\sqrt{\rho_k}\sqrt{\sigma_k} U_k}} \\
    &\leq \sumk \sqrt{\alpha_k \beta_k} \sumj \sqrt{p_{jk}}\sqrt{q_{jk}} \\
    &= \sumkj \sqrt{\rho_{j k}} \sqrt{\sigma_{j k}} \\
    &= r \cdot s,
\end{align*}
where the $U_k$ are some unitaries on $\mathcal{H}^A$. Thus, we see that
\begin{equation}\label{theta0andAngularDistance}
    \theta_0 = \arccos\parenth{r \cdot s} \leq \arccos\norm{\srss}_1 = \A{\rho}{\sigma},
\end{equation}
as needed.

Next, since the eigenvalues of $\rho$ and $\sigma$ completely determine the eigenvalues of their partial traces, it is possible to compute the conditional entropy of $\rho$ and $\sigma$ given only the vectors $r$ and $s$. To see this, consider the following function
\begin{equation*}\begin{aligned}
    \Hc \parenth{v} \coloneqq
    &- \sumkj v_{j k}^2 \ln v_{j k}^2 \\
    &+ \sumk \parenth{\sumj v_{j k}^2} \ln{\parenth{\suml v_{l k}^2}},
\end{aligned}\end{equation*}
where $v = \parenth{v_{j k}}_{k,j}$ can be any vector in $\mathbb{R}^d$. Then
\begin{equation}\begin{aligned}\label{HcrandHcs}
    \Hc(r) &= -\sumkj \rho_{j k} \ln \rho_{j k} + \sumk \parenth{\sumj \rho_{j k}} \ln\parenth{\suml \rho_{l k}} \\
    &= \HCond{\rho} \\
    \Hc(s) &= -\sumkj \sigma_{j k} \ln \sigma_{j k} + \sumk \parenth{\sumj \sigma_{j k}} \ln\parenth{\suml \sigma_{l k}} \\
    &= \HCond{\sigma},
\end{aligned}\end{equation}
so the vectors $r$ and $s$ are sufficient to determine the conditional entropies $\HCond{\rho}$ and $\HCond{\sigma}$.

The idea of our proof is now to integrate from $r$ to $s$ in $\mathbb{R}^d$, tracking the infinitesimal changes in the conditional entropy and angular distance. To see this formally, first note that $r$ and $s$ are unit vectors (with respect to the standard inner product on $\mathbb{R}^d$), since $r \cdot r = \tr{\rho} = 1 = \tr{\sigma} = s \cdot s$. Moreover, we have $r, s \geq 0$ by definition~(\ref{eq:rsdefn}). Now, note that if $r \cdot s = 1$, then $r = s$, so we have $\Hc(r) = \Hc(s)$ i.e.~$\HCond{\rho} = \HCond{\sigma}$. Since $u(d_A) \geq 0$ and $\A{\rho}{\sigma} \geq 0$, Eq.\ (\ref{ContinuityBound}) holds trivially in this case. Now consider the remaining case $r \cdot s \in [0, 1)$. Let $\tilde{s}$ be the normalized projection of $s$ onto the orthogonal complement of $\Span \set{r}$,
\begin{equation*}
    \tilde{s} \coloneqq \frac{s - \parenth{s \cdot r} r}{\abs{s - \parenth{s \cdot r} r}}.
\end{equation*}
Using $\tilde{s}$, we define the path
\begin{equation*}
    v(\theta) \coloneqq \cos(\theta) r + \sin(\theta) \tilde{s}
\end{equation*}
from $r$ to $s$, where $\theta \in [0, \theta_0]$. Note that $v(0) = r$, $v(\theta_0) = s$, and $v(\theta)$ traverses the great circle along the $(d-1)$-sphere from $r$ to $s$. In addition, note that $\abs{v(\theta)} = 1$ for all $\theta \in [0, \theta_0]$. Now, the tangent to the path $v(\theta)$ is
\begin{equation*}
    w(\theta) \coloneqq v'(\theta) = -\sin(\theta) r + \cos(\theta) \tilde{s},
\end{equation*}
which satisfies $\abs{w(\theta)} = 1$ and $v(\theta) \cdot w(\theta) = 0$ for all $\theta \in [0, \theta_0]$.

For notational simplicity, we now define $\Hc(\theta) \coloneqq \Hc(v(\theta))$, so
\begin{equation*}\begin{aligned}
    \Hc(\theta) = &-\sumkj \xxx{v}{j}{k}^2 \ln(\xxx{v}{j}{k}^2) \\ &+ \sumk \parenth{\sumj \xxx{v}{j}{k}^2} \ln{\parenth{\suml \xxx{v}{l}{k}^2}}.
\end{aligned}\end{equation*}
Observe that $\Hc(\theta)$ is continuous on $[0,\theta_0]$ (under the standard convention for entropy definitions that $0 \ln 0 \equiv 0$).
Thus, if we show that $\Hc(\theta)$ is differentiable on $(0,\theta_0)$ and its derivative satisfies $\abs{\Hc'(\theta)} \leq u(d_A)$ on that interval, then the desired result follows immediately, since
\begin{equation*}\begin{aligned}
    \abs{\HCond{\sigma} - \HCond{\rho}} &= \abs{\Hc(\theta_0) - \Hc(0)} \\
    &= \abs{\int_{0}^{\theta_0} \Hc'(\theta) \, d\theta} \\
    &\leq \int_{0}^{\theta_0} \abs{\Hc'(\theta)} \, d\theta \\
    &\leq u(d_A) \theta_0 \\
    &\leq u(d_A) \A{\rho}{\sigma},
\end{aligned}\end{equation*}
where the first line follows from Eq.~(\ref{HcrandHcs}), and the last line follows from Eq.~(\ref{theta0andAngularDistance}).
Thus, all that remains is to bound $\abs{\Hc'(\theta)}$ by $u(d_A)$.

To do this, we first handle a technicality regarding zero eigenvalues. For each $k \in \set{1, \ldots, d_B}$, let $S^k_A$ be the set of $j \in \set{1, \ldots, d_A}$ such that at least one of $r_{j k},s_{j k}$ is nonzero. Furthermore, let $S_B$ be the set of $k \in \set{1, \ldots, d_B}$ such that $S^k_A$ is nonempty. Then for any $(k,j)$ with $k\in S_B$ and $j\in S^k_A$, at least one of $r_{j k},s_{j k}$ is nonzero, which implies that $\xxx{v}{j}{k}^2 > 0$ for all $\theta \in (0, \theta_0)$. Moreover, for all other $(k,j)$, we have that $r_{j k}=s_{j k}=0$, so $\xxx{v}{j}{k}^2=0$ for all $\theta \in (0, \theta_0)$, which implies that the value of $\Hc(\theta)$ would not change upon removing the term $\xxx{v}{j}{k}^2$. Thus, in the remainder of the argument, summations of the form $\sum_{k,j}$ should be understood to mean $\sum_{k\in S_B} \sum_{j\in S^k_A}$ (and analogously, $\sum_{l}$ means $\sum_{l\in S^k_A}$), which ensures that all terms appearing in the summations satisfy $\xxx{v}{j}{k}^2 > 0$ and $\sum_{l\in S^k_A} \xxx{v}{l}{k}^2 > 0$ for all $\theta \in (0, \theta_0)$.
With this, we see that $\Hc(\theta)$ is indeed differentiable on $(0, \theta_0)$, and

\begin{equation*}\begin{aligned}
    \abs{\Hc'(\theta)} = \Bigg| \! &- \sum_{k,j} 2\xxx{v}{j}{k}\xxx{w}{j}{k} \ln(\xxx{v}{j}{k}^2) \\
    &- \sum_{k,j} \xxx{v}{j}{k}^2 \frac{2\xxx{v}{j}{k}\xxx{w}{j}{k}}{\xxx{v}{j}{k}^2} \\
    &+ \sum_{k,j} 2\xxx{v}{j}{k}\xxx{w}{j}{k} \ln{\parenth{\sum_l \xxx{v}{l}{k}^2}} \\
    &+ \sum_{k,j} \xxx{v}{j}{k}^2 \frac{\sum_l 2\xxx{v}{l}{k}\xxx{w}{l}{k}}{\sum_l \xxx{v}{l}{k}^2} \Bigg|.
\end{aligned}\end{equation*}
Using $v(\theta) \cdot w(\theta) = 0$, this simplifies to
\begin{equation}\begin{aligned}\label{HcThetaBoundIntermediate}
    \abs{\Hc'(\theta)}
    &= 2 \abs{\sum_{k,j} \xxx{v}{j}{k}\xxx{w}{j}{k} \ln\parenth{\frac{\sum_l \xxx{v}{l}{k}^2}{\xxx{v}{j}{k}^2}}} \\
    &\leq 2 \sqrt{\sum_{k,j} \xxx{v}{j}{k}^2 \ln^2\parenth{\frac{\sum_l \xxx{v}{l}{k}^2}{\xxx{v}{j}{k}^2}}},
\end{aligned}\end{equation}
where in the last line we used the Cauchy--Schwarz inequality with $\abs{w(\theta)} = 1$.

Now, recall that the $\xxx{v}{j}{k}^2$ form a valid probability distribution (i.e.\ they are non-negative values summing to $1$), since $\abs{v}=1$. Also, note that the argument of $\ln^2$ in the final line above, i.e.\ $\frac{\sum_l \xxx{v}{l}{k}^2}{\xxx{v}{j}{k}^2}$, lies in the interval $[1,\infty)$. Thus, we now construct an increasing concave upper bound $f(x)$ for $\ln^2 x$ on $x \in [1,\infty)$, as this would allow us to ``move the summation" over $k,j$ (weighted by the probabilities $\xxx{v}{j}{k}^2$) into the argument of the function. To begin, note that
\begin{equation*}\begin{aligned}
    \frac{d}{dx} \ln^2 x &= 2 \frac{\ln x}{x} \\
    \frac{d^2}{dx^2} \ln^2 x &= 2 \frac{1-\ln x}{x^2},
\end{aligned}\end{equation*}
so $\ln^2 x$ is convex for all $x \in [1,e]$ and concave for all $x \in [e, \infty)$. Then to produce $f(x)$, we seek a line $y(x) = m(x-a)$ such that $y(1) = \ln^2(1) = 0$, and such that there exists $x_0 \in [1, \infty)$ with $x_0 \geq e$, $y(x_0) = \ln^2(x_0)$ and $y'(x_0) = \parenth{\frac{d}{dx} \ln^2 x}_{x_0} = 2 \frac{\ln x_0}{x_0}$. Then we must solve the system
\begin{equation*}
    \begin{aligned}
        \ln^2 x_0 = y(x_0) = 2\frac{\ln x_0}{x_0}(x_0-1),
    \end{aligned}
\end{equation*}
for $x_0$, which has solutions $x_0 = 1$ and $x_0 = \exp{\parenth{\lambertW{-\frac{2}{e}}}} \approx 4.922$. We discard the solution $x_0 = 1 < e$ and keep the other solution $x_0 \approx 4.922 \geq e$. Thus, our increasing concave upper bound for $\ln^2 x$ is
\begin{equation*}
    f(x) \coloneqq
    \begin{cases}
        2\frac{\ln{x_0}}{x_0}(x-1) & 1 \leq x \leq x_0 \\
        \ln^2 x & x \geq x_0
    \end{cases}.
\end{equation*}
With $f(x)$, we can write
\begin{equation}\label{HcThetaBoundFinal}
\begin{aligned}
    \abs{\Hc'(\theta)} &\leq 2 \sqrt{\sum_{k,j} \xxx{v}{j}{k}^2 \ln^2\parenth{\frac{\sum_l \xxx{v}{l}{k}^2}{\xxx{v}{j}{k}^2}}} \\
    &\leq 2 \sqrt{\sum_{k,j} \xxx{v}{j}{k}^2 f\parenth{\frac{\sum_l \xxx{v}{l}{k}^2}{\xxx{v}{j}{k}^2}}} \\
    &\leq 2 \sqrt{f\parenth{\sum_{k,j} \xxx{v}{j}{k}^2 \frac{\sum_l \xxx{v}{l}{k}^2}{\xxx{v}{j}{k}^2}}} \\
    &\leq 2 \sqrt{f\parenth{d_A \sum_{k,l} \xxx{v}{l}{k}^2}} \\
    &= 2 \sqrt{f(d_A)} \\
    &= \LipschitzConstant{d_A} \\
    &= u(d_A).
\end{aligned}
\end{equation}
In the above, the first line is Eq.~(\ref{HcThetaBoundIntermediate}). The second line follows since $f$ is an upper bound on $\ln^2 x$ for $x \geq 1$. The third line follows since $f$ is concave and $\abs{v(\theta)} = 1$. The fourth line follows since $f$ is increasing and $\abs{S_A^k} \leq d_A$ for all $k$. The fifth line follows since $\abs{v(\theta)} = 1$. The sixth and seventh lines follow from the definitions and the fact that $d_A \in \mathbb{N}$.
\end{proof}

In comparison to the proof in~\cite{SBV+21} for unconditioned entropies, the main difference in our proof here is essentially that there are additional contributions to the derivative $\Hc'(\theta)$ arising from the $\Hvn{\rho_B}$ term in the conditional entropy. Informally, these contributions act in the ``opposite direction'' from those of the $\Hvn{\rho_{AB}}$ term, reducing the magnitude of the derivative and yielding a final bound that is independent of $d_B$, in contrast to what we would have obtained had we only considered the derivative of the $\Hvn{\rho_{AB}}$ term alone --- see Eq.\ (\ref{SekatskiBoundConditional}).
Another small difference is that we have constructed the concave upper bound in a slightly different and arguably simpler way.

\subsection{Potential improvements}

How tight is the bound in Eq.~(\ref{ContinuityBound})? To address this question empirically, we began by randomly sampling $100,000$ pairs of quantum-classical states according to the procedure in Appendix~\ref{RandomQuantumClassicalStates}, and for each pair $(\rho, \sigma)$, we computed their angular distance $\A{\rho}{\sigma}$ and conditional entropy difference $\abs{\HCond{\rho} - \HCond{\sigma}}$. We then plotted the conditional entropy differences against the angular distances in Fig.~\ref{fig:RandomQuantumClassicalStates}. This provides an empirical estimate for the tightness of Eq.~(\ref{ContinuityBound}) at all feasible angular distances $\A{\rho}{\sigma} \in \left[0, \frac{\pi}{2} \right]$.

\begin{figure}
\includegraphics[width = 0.48 \textwidth]{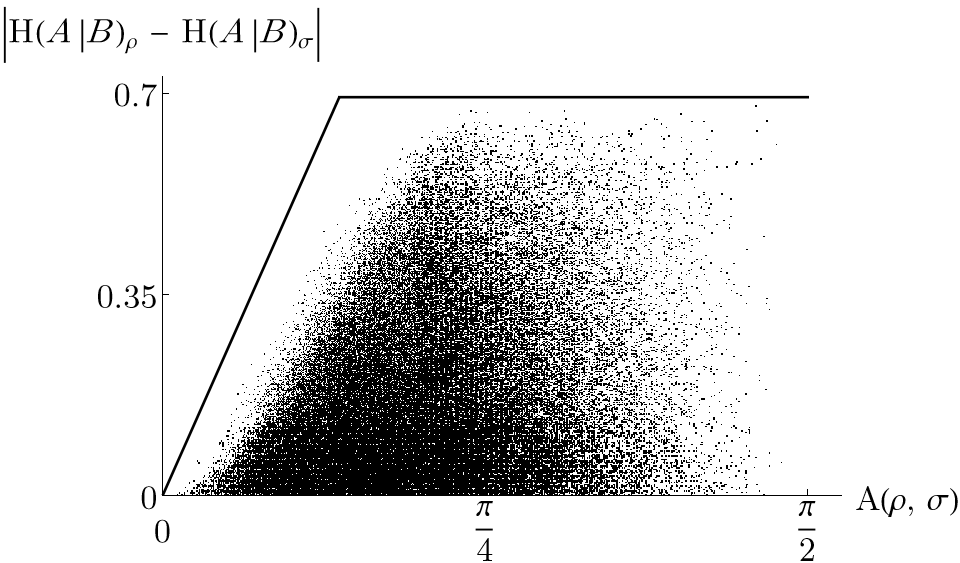}
\caption{Scatterplot of the conditional entropy difference $\abs{\HCond{\rho} - \HCond{\sigma}}$ against the angular distance $\A{\rho}{\sigma}$ for $100,000$ pairs of randomly sampled quantum-classical states $(\rho, \sigma)$ with $d_A = 2 = d_B$ (see Appendix~\ref{RandomQuantumClassicalStates}). 
For comparison, we have also plotted the curve $\min\{u(d_A) \A{\rho}{\sigma}, \ln d_A\}$, where $u(d_A) \A{\rho}{\sigma}$ is the bound in Eq.~(\ref{ContinuityBound}), and $\ln d_A$ is a hard upper bound on the conditional entropy difference between quantum-classical states.}
\label{fig:RandomQuantumClassicalStates}
\end{figure}

Next, we explore the tightness of Eq.~(\ref{ContinuityBound}) at small angular distances (since in most applications we are mainly interested in this case). For each angular distance $\A{\rho}{\sigma} \in \set{0.1 \times 10^{-5}, 0.2 \times 10^{-5}, \ldots, 1.0 \times 10^{-5}}$, we randomly sampled $10,000$ pairs of classical states $(\rho, \sigma)$ with that angular distance, as described in Appendix~\ref{RandomClassicalStatesFixedAngularDistance}. Then, we computed the conditional entropy differences $\abs{\HCond{\rho} - \HCond{\sigma}}$ and plotted them in Fig.~\ref{fig:FixedAngularDistance}. The results suggest that at small angular distances, our continuity bound is close to the ``true" tight expression when $d_A$ is small, but there may be room for improvement when $d_A$ is larger (note that random sampling typically yields less representative results in high dimensions, so the latter claim should not be taken as conclusive).

\begin{figure}
\includegraphics[width=0.48\textwidth]{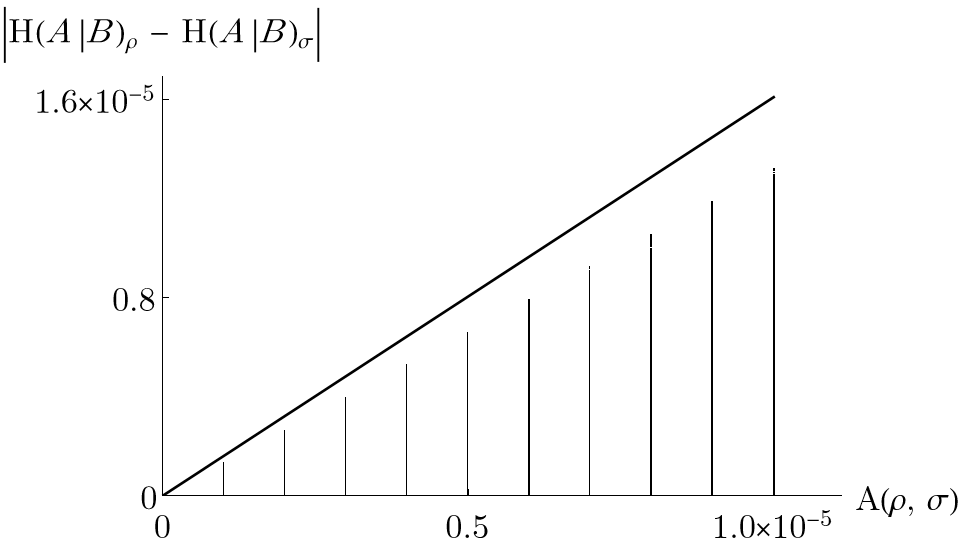}\\
\includegraphics[width=0.48\textwidth]{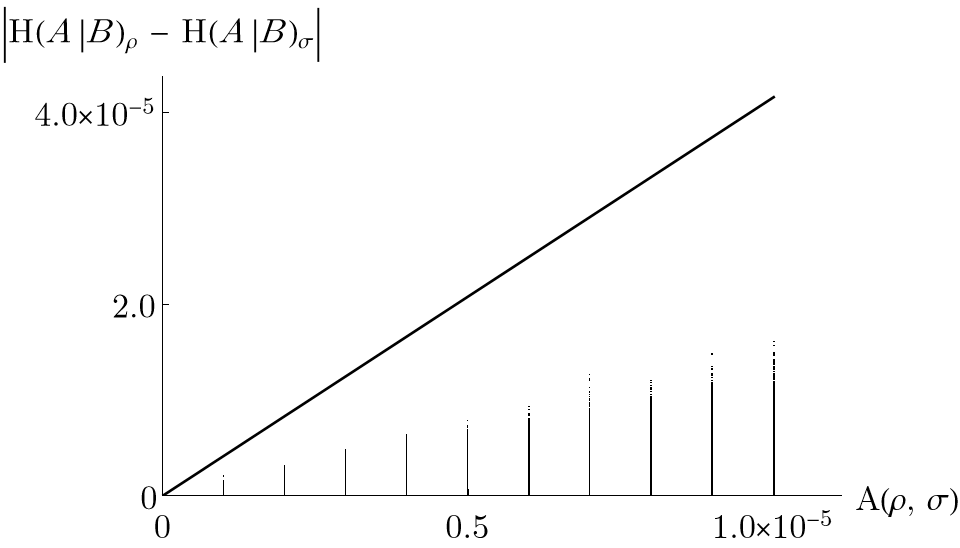}
\caption{Scatterplot of the conditional entropy difference $\abs{\HCond{\rho} - \HCond{\sigma}}$ against the angular distance $\A{\rho}{\sigma}$ for $100,000$ pairs of classical states $(\rho, \sigma)$ randomly sampled at fixed angular distances $\A{\rho}{\sigma} \in \set{0.1 \times 10^{-5}, 0.2 \times 10^{-5}, \ldots, 1.0 \times 10^{-5}}$ (see Appendix~\ref{RandomClassicalStatesFixedAngularDistance}). The slanted line shows the bound $u(d_A) \A{\rho}{\sigma}$ in Eq.~(\ref{ContinuityBound}), while the vertical ``lines'' are formed by the data points from the $100,000$ sampled pairs $(\rho, \sigma)$. Top panel: $d_A = 2 = d_B$; bottom panel: $d_A = 8$, $d_B = 2$.}
\label{fig:FixedAngularDistance}
\end{figure}

Note that in order to saturate Eq.\ (\ref{ContinuityBound}), a pair of states $(\rho, \sigma)$ must saturate both inequalities~\eqref{HcThetaBoundIntermediate} (Cauchy--Schwarz) and \eqref{HcThetaBoundFinal} (which roughly speaking is due to the concavity of $f$). However, it seems that these inequalities cannot be simultaneously saturated, which is consistent with the above empirical evidence that there is room for sharpening the bound.

It is also worth briefly comparing ``conversions'' between our result and the continuity bounds based on trace distance. Specifically, note that the Fuchs--van de Graaf inequalities~\cite{FvdG99} upper bound the trace distance in terms of angular distance and vice versa. Thus, a continuity bound in terms of either distance measure in principle yields a continuity bound in terms of the other. However, such a conversion is potentially quite suboptimal --- we provide a brief scaling comparison in Appendix~\ref{app:compare}, where we find that starting from a bound on angular distance and then applying the previous continuity bound~\eqref{WinterBound} yields highly suboptimal results, while the other direction (starting from a bound on trace distance and then applying our bound~\eqref{ContinuityBound}) is somewhat better, though still not tight.

Finally, we discuss avenues for generalizing Eq.~(\ref{ContinuityBound}), since many applications in quantum key distribution require Eq.~(\ref{ContinuityBound}) (or a similar bound satisfying both conditions \ref{DesiredBoundConditionIndependence} and \ref{DesiredBoundConditionDifferentiable}) to hold for classical-quantum states as well. However, our proof technique does not appear to generalize readily to classical-quantum (or fully quantum) states, since our proof relies on the simple eigenvalue relationship between a quantum-classical state $\rho$ and its partial trace $\trA{\rho}$. This eigenvalue relationship becomes much more complicated in the classical-quantum (or fully quantum) case, since the eigenvalues of $\trA{\rho}$ now depend on the eigenvectors of $\rho$ as well. We also highlight that as observed in~\cite{KA07}, it seems difficult to use purification-based arguments to obtain such a result, because purifications usually do not ``preserve'' the conditional entropies in a useful way --- we give some further details in Appendix~\ref{app:purification}.

In attempting to generalize Eq.~(\ref{ContinuityBound}), it is important to note that the bound does not hold for arbitrary (fully quantum) states $\rho, \sigma \in \mathcal{D}\parenth{\mathcal{H}^A \otimes \mathcal{H}^B}$. To see this, let $\set{\ket{e_j}}_j$ and $\set{\ket{f_k}}_k$ be orthonormal bases for $\mathcal{H}^A$ and $\mathcal{H}^B$, respectively. Let $d_M \coloneqq \min\set{d_A, d_B}$, and consider the maximally entangled state
\begin{equation}\label{rhoCounterexample}
    \rho = \frac{1}{d_M} \sum_{j,k=1}^{d_M} \ket{e_j f_j} \bra{e_k f_k},
\end{equation}
which has the most negative conditional entropy $\HCond{\rho} = - \ln d_M$. Next, consider the maximally mixed state
\begin{equation*}
    \tau = \frac{1}{d_A d_B} \mathbbm{1},
\end{equation*}
which has the most positive conditional entropy $\HCond{\tau} = \ln d_A$. Now, consider the segment connecting $\rho$ and $\tau$,
\begin{equation}\label{sigmaCounterexample}
    \sigma = \lambda \, \tau + \parenth{1-\lambda} \rho,
\end{equation}
where $\lambda \in [0,1]$. By direct computation, one can show that
\begin{equation*}
    \A{\rho}{\sigma} = \arccos\parenth{\sqrt{1 - \frac{d_A d_B - 1}{d_A d_B} \lambda}},
\end{equation*}
and that
\begin{equation*}\begin{aligned}
    \Big|\HCond{\rho} &- \HCond{\sigma} \Big| = - \ln d_M \\
    &+ \left(1 - \frac{d - 1}{d} \lambda \right) \ln \left(1 - \frac{d - 1}{d} \lambda \right) \\
    &+ (d-1) \frac{\lambda}{d} \ln \frac{\lambda}{d} \\
    &- (d_B-d_M) \frac{\lambda}{d_B} \ln \frac{\lambda}{d_B} \\
    &- d_M \left(\frac{\lambda }{d_B}+\frac{1 - \lambda}{d_M}\right) \ln \left(\frac{\lambda }{d_B}+\frac{1 - \lambda}{d_M}\right),
\end{aligned}\end{equation*} 
where $d \coloneqq d_A d_B$. For the counterexample to Eq.~(\ref{ContinuityBound}), let $d_A = 2 = d_B$ and $\lambda = \frac{1}{2}$. Then
\begin{equation*}
    \abs{\HCond{\rho} - \HCond{\sigma}} \approx 1.074,
\end{equation*}
but
\begin{equation*}
    u(d_A) \A{\rho}{\sigma} \approx 1.061,
\end{equation*}
in violation of Eq.~(\ref{ContinuityBound}). This situation is depicted in the top panel of Fig.~\ref{fig:Counterexample}.

\begin{figure}
\includegraphics[width=0.48\textwidth]{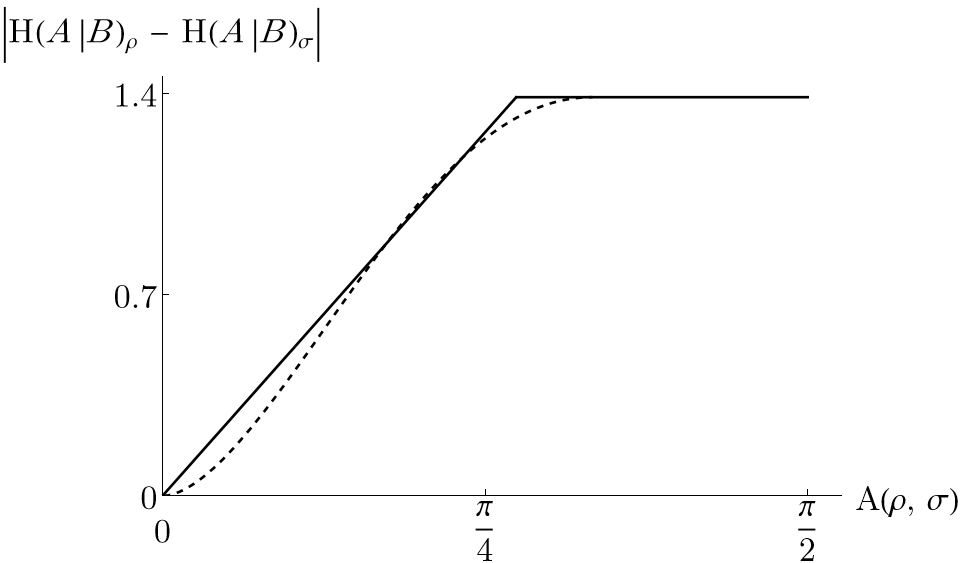}\\
\includegraphics[width=0.48\textwidth]{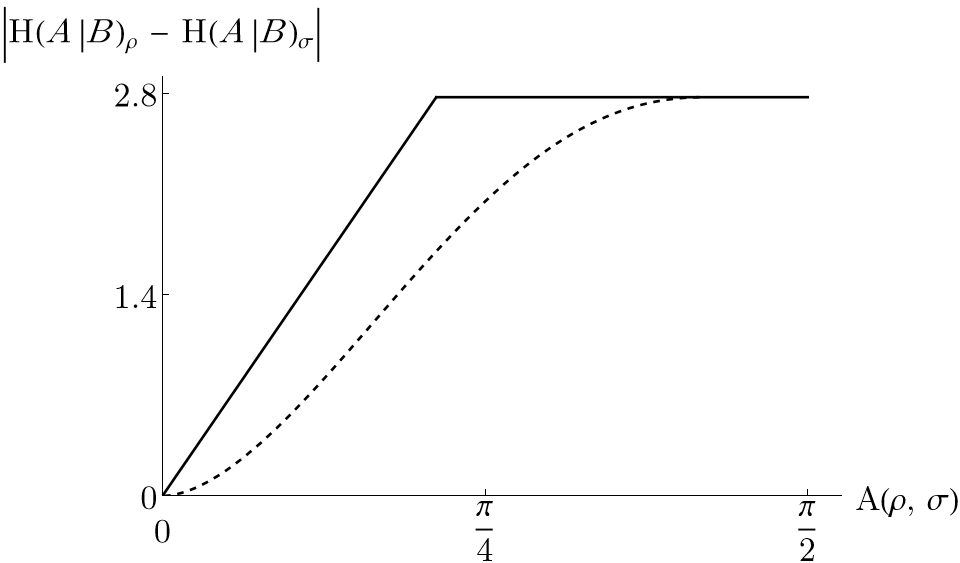}
\caption{Parametric plot (dashed line) of $\abs{\HCond{\rho} - \HCond{\sigma}}$ against $\A{\rho}{\sigma}$, where $\rho$ and $\sigma$ are as in Eqs.\ (\ref{rhoCounterexample}) and (\ref{sigmaCounterexample}) respectively, for $\lambda \in [0,1]$. For comparison, the solid line shows $\min\{u(d_A) \A{\rho}{\sigma}, \, \ln (d_A d_M) \}$ with $d_M \coloneqq \min \set{d_A, d_B}$, where $u(d_A) \A{\rho}{\sigma}$ is the bound in Eq.~(\ref{ContinuityBound}), and $\ln (d_A d_M)$ is a hard upper bound on the conditional entropy difference between fully quantum states. Top panel: $d_A = 2 = d_B$; bottom panel: $d_A = 8$, $d_B = 2$.}
\label{fig:Counterexample}
\end{figure}

However, it appears that violations of Eq.~(\ref{ContinuityBound}) are uncommon and relatively small in magnitude --- the above is the only and the most egregious counterexample to Eq.~(\ref{ContinuityBound}) known to the authors. Also, numerical computation shows that for every choice of $d_A \in \set{2,\ldots,10}$, every choice of $d_B \in \set{1,\ldots,10}$, and for $\lambda \in [0,1]$, the above construction produces a counterexample only when $d_A = 2 = d_B$ and $\lambda \approx 0.5$ (see the bottom panel of Fig.~\ref{fig:Counterexample} for a representative example). Thus, a slight modification of the bound in Eq.~(\ref{ContinuityBound}), perhaps by including an extra factor of $2$ (i.e.\ $u(d) \to 2 \, u(d)$ in Eq.~(\ref{LipschitzConstant})), will resolve the only counterexample known to the authors and may generalize the bound in Eq.~(\ref{ContinuityBound}) to all states $\rho, \sigma \in \mathcal{D}\parenth{\mathcal{H}^A \otimes \mathcal{H}^B}$.

\section{Trace Distance and Angular Distance}\label{sec:FvdG}

We now shift our attention to the relationship between the trace distance and the angular distance. Specifically, we study the set of states which saturate the Fuchs--van de Graaf inequalities \cite{FvdG99}
\begin{equation}\label{QuantumFvdGInequalities}
    1 - \F{\rho}{\sigma} \leq \T{\rho}{\sigma} \leq \sqrt{1 - \F{\rho}{\sigma}^2}.
\end{equation}
To motivate this direction of investigation, note that Audenaert's continuity bound for the von Neumann entropy in Eq.~(\ref{AudenaertBound}) is tight \cite{KA07}. That is, for any $T \in [0,1]$, there exist $\rho, \sigma \in \mathcal{D}(\mathcal{H})$ such that $\T{\rho}{\sigma} = T$ and $\abs{\Hvn{\rho} - \Hvn{\sigma}} = T \ln\parenth{d-1} + \hbinary{T}$. But suppose that some such $\rho, \sigma$ which saturate Eq.~(\ref{AudenaertBound}) also saturate the right-hand side of Eq.~(\ref{QuantumFvdGInequalities}) (which we refer to as the ``upper Fuchs--van de Graaf inequality"), i.e.
\begin{equation*}
    \T{\rho}{\sigma} = \sqrt{1 - \F{\rho}{\sigma}^2} = \sin \A{\rho}{\sigma}.
\end{equation*}
Then for $\T{\rho}{\sigma} \ll 1$, we have $\T{\rho}{\sigma} = \sin \A{\rho}{\sigma} \approx \A{\rho}{\sigma} \eqqcolon A$. But then it is impossible to obtain a continuity bound for the conditional entropy which satisfies both condition~\ref{DesiredBoundConditionIndependence} and condition~\ref{DesiredBoundConditionDifferentiable}, since whenever $d_B = 1$, such a pair $\rho, \sigma$ would satisfy 
\begin{equation*}\begin{aligned}
    \abs{\HCond{\rho} - \HCond{\sigma}} &= \abs{\Hvn{\rho} - \Hvn{\sigma}} \\
    &= T \ln\parenth{d-1} + \hbinary{T} \\
    &\approx A \ln\parenth{d-1} + \hbinary{A},
\end{aligned}\end{equation*}
so the conditional entropy scales badly with the angular distance at $A \approx 0$, and condition~\ref{DesiredBoundConditionDifferentiable} would not be satisfiable.

\subsection{Invertible states}\label{sec:InvertibleStates}

With this motivation, our goal is now to characterize the set of all states which saturate the Fuchs--van de Graaf inequalities. We begin by studying the subset of such states that are also invertible (i.e.\ positive definite), since this case is easier to handle.
That is, we characterize the sets
\begin{equation}\label{S1Def}
    \mathcal{S}_1 \coloneqq \set{(\rho, \sigma) \in \mathcal{D}_{\text{inv}}(\mathcal{H})^2 \mid 1 - \F{\rho}{\sigma} = \T{\rho}{\sigma}}
\end{equation}
and
\begin{equation}\label{S2Def}
    \mathcal{S}_2 \coloneqq \set{(\rho, \sigma) \in \mathcal{D}_{\text{inv}}(\mathcal{H})^2 \biggm| \T{\rho}{\sigma} = \sqrt{1 - \F{\rho}{\sigma}^2}},
\end{equation}
where $\mathcal{D}_{\text{inv}}(\mathcal{H})$ denotes the set of all invertible density operators on $\mathcal{H}$. To do this, we apply the following general strategy. A generic inequality $\G_1(\rho, \sigma) \leq \G_n(\rho, \sigma)$ is usually proven via a chain of inequalities
\begin{equation*}
    \G_1(\rho, \sigma) \leq \G_2(\rho, \sigma) \leq \ldots \leq \G_{n-1}(\rho, \sigma) \leq  \G_n(\rho, \sigma).
\end{equation*}
Thus, to determine which states $\rho, \sigma$ satisfy $\G_1(\rho, \sigma) = \G_n(\rho, \sigma)$, we derive the equality conditions for each inequality $\G_j(\rho, \sigma) \leq \G_{j+1}(\rho, \sigma)$ in the above chain of inequalities, and we combine the equality conditions (with logical and) for all $j \in \set{1, \ldots, n-1}$ to obtain the equality condition for $\G_1(\rho, \sigma) \leq \G_n(\rho, \sigma)$. However, there usually exist multiple different proofs of the inequality $\G_1(\rho, \sigma) \leq \G_n(\rho, \sigma)$, which each proceed through alternative chains of intermediate inequalities
\begin{equation*}
    \G_1(\rho, \sigma) \leq \tilde{\G}_2(\rho, \sigma) \leq \ldots \leq \tilde{\G}_{n-1}(\rho, \sigma) \leq  \G_n(\rho, \sigma).
\end{equation*}
Thus, a careful consideration of the various proof techniques for a given inequality is needed for analyzing equality conditions, since different proof techniques may yield equality conditions vastly different in appearance. Of course, all equality conditions for the same inequality $\G_1(\rho, \sigma) \leq \G_n(\rho, \sigma)$ should be logically equivalent regardless of the underlying proof technique, but some equality conditions may not be ``compatible" with others in our overall proof, making it more difficult to condense all the intermediate equality conditions into a final concise characterization.

To proceed, we need to introduce a few definitions. First, for any probability distributions $p$ and $q$ on a finite alphabet $\mathcal{X}$, we denote their classical trace distance by
\begin{equation*}
    \Tc{p}{q} \coloneqq \frac{1}{2} \sum_{\xinX} \abs{p(x) - q(x)}
\end{equation*}
and their classical fidelity by
\begin{equation*}
    \Fc{p}{q} \coloneqq \sum_{\xinX} \sqrt{p(x)q(x)},
\end{equation*}
which are special cases of Eqs.\ (\ref{TraceDistanceDef}) and (\ref{FidelityDef}) for commuting density operators. Next, we denote the set of all rank-1 projective measurements on a Hilbert space $\mathcal{H}$ by
\begin{equation*}\begin{aligned}
    \mathcal{M} \coloneqq \Big\{ \! &\set{\ketbra{e_x}{e_x}}_{\xinX} \subseteq \mathcal{L}(\mathcal{H}) \Bigm| \\
    &\set{\ket{e_x}}_{\xinX} \textup{ is an orthonormal basis for } \mathcal{H} \Big\},
\end{aligned}\end{equation*}
where $\mathcal{L}(\mathcal{H})$ denotes the space of linear operators on $\mathcal{H}$. Now, each $\Lambda \coloneqq \set{\ketbra{e_x}{e_x}}_{\xinX} \in \mathcal{M}$ and $\tau \in \mathcal{D}(\mathcal{H})$ induce a natural probability distribution
\begin{equation*}
    \tr{\Lambda \tau} \coloneqq \parenth{\tr{\ketbra{e_x}{e_x} \tau}}_{\xinX} = \parenth{\braket{e_x}{\tau \, e_x}}_{\xinX} \in \mathbb{R}^d,
\end{equation*}
where $d \coloneqq \dim \mathcal{H}$. That is, this probability distribution is given by
\begin{equation*}
    \tr{\Lambda \tau}(x) \coloneqq \tr{\ketbra{e_x}{e_x} \tau} = \braket{e_x}{\tau \, e_x}
\end{equation*}
for all $\xinX$. Finally, for any positive definite operators $A$ and $B$ on a Hilbert space $\mathcal{H}$, we denote their geometric mean \cite{Ando04} by
\begin{equation}\label{eq:geomean}
    A \# B \coloneqq \power{A}{} \sqrt{\power{A}{-} B \power{A}{-}} \power{A}{} = B \# A.
\end{equation}

Now, to obtain equality conditions for Eq.~(\ref{QuantumFvdGInequalities}), we found that Fuchs and van de Graaf's original proof technique in \cite{FvdG99} seemed amenable for analysis. Their proof of Eq.~(\ref{QuantumFvdGInequalities}) invoked a variational characterization of the trace distance, a variational characterization of the fidelity, and a classical version of the bound in Eq.~(\ref{QuantumFvdGInequalities}). We first introduce these previously known results in turn. Then, we derive new equality conditions for these results as lemmas. Finally, we combine the lemmas in Theorem \ref{QuantumFvdGTheorem} to obtain characterizations of the sets $\mathcal{S}_1$ and $\mathcal{S}_2$ in Eqs.\ (\ref{S1Def}) and (\ref{S2Def}). Proofs of the lemmas have been deferred to Appendix~\ref{app:proofsFvdG} for readability.

We first consider the variational characterization of the trace distance, which is stated in \cite{FvdG99} and can be traced back to the work of Helstrom and Toussaint \cite{Toussaint72,Helstrom76}. For any $\rho, \sigma \in \mathcal{D}(\mathcal{H})$, their trace distance can be expressed as
\begin{equation}\label{TraceDistanceOptimization}
    \T{\rho}{\sigma} = \max_{\Lambda \in \mathcal{M}} \Tc{\tr{\Lambda \rho}}{\tr{\Lambda \sigma}}.
\end{equation}
We define
\begin{equation*}
    \mathcal{T}(\rho, \sigma) \coloneqq \set{\Lambda \in \mathcal{M} \mid \T{\rho}{\sigma} = \Tc{\tr{\Lambda \rho}}{\tr{\Lambda \sigma}}}
\end{equation*}
to be the (nonempty) set of all rank-1 projective measurements which achieve the optimum in Eq.~(\ref{TraceDistanceOptimization}). We characterize this set in the following lemma.

\begin{lemma}\label{VariationalTraceDistanceLemma}

For any $\rho, \sigma \in \mathcal{D}(\mathcal{H})$, we have
\begin{equation*}\begin{aligned}
    \mathcal{T}(\rho, \sigma) = \Big\{ &\Lambda \coloneqq \set{\ketbra{e_x}{e_x}}_{\xinX} \in \mathcal{M} \Bigm| \\ &\forall \xinX,\ket{e_x} \in \ker P \lor \ket{e_x} \in \ker Q \Big\} \neq \emptyset,
\end{aligned}\end{equation*}
where $P$ and $Q$ are the positive and negative parts of $\rho - \sigma$, respectively.

\end{lemma}

\begin{proof}
Appendix~\ref{ProofVariationalTraceDistanceLemma}.
\end{proof}

Next, we consider the variational characterization of the fidelity, which was first introduced by Fuchs and Caves in \cite{FC95}. For any $\rho, \sigma \in \mathcal{D}(\mathcal{H})$, their fidelity can be expressed as
\begin{equation}\label{FidelityOptimization}
    \F{\rho}{\sigma} = \min_{\Lambda \in \mathcal{M}} \Fc{\tr{\Lambda \rho}}{\tr{\Lambda \sigma}}.
\end{equation}
We define
\begin{equation*}
    \mathcal{F}(\rho, \sigma) \coloneqq \set{\Lambda \in \mathcal{M} \mid \F{\rho}{\sigma} = \Fc{\tr{\Lambda \rho}}{\tr{\Lambda \sigma}}}
\end{equation*}
to be the (nonempty) set of all rank-1 projective measurements which achieve the optimum in Eq.~(\ref{FidelityOptimization}). We characterize this set for invertible $\rho$ and $\sigma$ in the following lemma.

\begin{lemma}\label{VariationalFidelityLemma}

For any invertible $\rho, \sigma \in \mathcal{D}_{\textup{inv}}(\mathcal{H})$, we have
\begin{equation*}\begin{aligned}
    \mathcal{F}(\rho, \sigma) = \Big\{ &\Lambda \coloneqq \set{\ketbra{e_x}{e_x}}_{\xinX} \in \mathcal{M} \Bigm| \\ &\set{\ket{e_x}}_{\xinX} \textup{ is an eigenbasis for } M \Big\} \neq \emptyset,
\end{aligned}\end{equation*}
where
\begin{equation*}
    M \coloneqq \rho^{-1} \# \sigma = \M
\end{equation*}
is the operator geometric mean between $\rho^{-1}$ and $\sigma$.

\end{lemma}

\begin{proof}
Appendix~\ref{ProofVariationalFidelityLemma}.
\end{proof}

Finally, we consider the classical analogue of Eq.~(\ref{QuantumFvdGInequalities}). As shown in \cite{FvdG99}, for any probability distributions $p$ and $q$ on a finite alphabet $\mathcal{X}$, the classical trace distance and classical fidelity are related by
\begin{equation}\label{ClassicalFvdGInequalities}
    1 - \Fc{p}{q} \leq \Tc{p}{q} \leq \sqrt{1 - \Fc{p}{q}^2}.
\end{equation}
We define
\begin{equation*}\begin{aligned}
    \mathcal{C}_1(\rho, \sigma) \coloneqq \Big\{ \Lambda \in \mathcal{M} \Bigm| \, 1 &- \Fc{\tr{\Lambda \rho}}{\tr{\Lambda \sigma}} \\
    &= \Tc{\tr{\Lambda \rho}}{\tr{\Lambda \sigma}} \Big\}
\end{aligned}\end{equation*}
and
\begin{equation*}\begin{aligned}
    \mathcal{C}_2(\rho, \sigma) \coloneqq \bigg\{ \Lambda \in \mathcal{M} \biggm| \, &\Tc{\tr{\Lambda \rho}}{\tr{\Lambda \sigma}} \\
    &= \sqrt{1 - \Fc{\tr{\Lambda \rho}}{\tr{\Lambda \sigma}}^2}\bigg\}
\end{aligned}\end{equation*}
to be the sets of rank-1 projective measurements which induce classical probability distributions $p \coloneqq \tr{\Lambda \rho}$ and $q \coloneqq \tr{\Lambda \sigma}$ which saturate the left-hand and right-hand inequalities in Eq.~(\ref{ClassicalFvdGInequalities}), respectively. We characterize these sets in the following lemma.

\begin{lemma}\label{ClassicalFvdGInequalitiesLemma}

For any $\rho, \sigma \in \mathcal{D}(\mathcal{H})$, we have
\begin{equation*}\begin{aligned}
    \mathcal{C}_1(\rho, \sigma) = \Big\{ &\Lambda \coloneqq \set{\ketbra{e_x}{e_x}}_{\xinX} \in \mathcal{M} \Bigm| \\ &\forall \xinX, \braket{e_x}{\rho \, e_x} = \braket{e_x}{\sigma \, e_x} \\
    &\lor \braket{e_x}{\rho \, e_x} = 0 \lor \braket{e_x}{\sigma \, e_x} = 0 \Big\}
\end{aligned}\end{equation*}
and
\begin{equation*}\begin{aligned}
    \mathcal{C}_2(\rho, \sigma) &= \\
    \Bigg\{ \Lambda &\coloneqq \set{\ketbra{e_x}{e_x}}_{\xinX} \in \mathcal{M} \Biggm| \\
    &\tr{\Lambda \rho} = \tr{\Lambda \sigma} \lor \tr{\Lambda \rho} \cdot \tr{\Lambda \sigma} = 0 \, \lor \\ &\parenth{\begin{aligned} &\exists b \in (0,1), \forall \xinX, \braket{e_x}{\sigma \, e_x} = b \braket{e_x}{\rho \, e_x} \\ &\lor \braket{e_x}{\sigma \, e_x} = \frac{1}{b} \braket{e_x}{\rho \, e_x} \end{aligned}} \Bigg\}.
\end{aligned}\end{equation*}

\end{lemma}

\begin{proof}
Appendix~\ref{ProofClassicalFvdGInequalitiesLemma}.
\end{proof}

We are now ready to state the main result of this section: equality conditions for the left-hand and right-hand inequalities in Eq.~(\ref{QuantumFvdGInequalities}). Intuitively, Theorem \ref{QuantumFvdGTheorem} says that for invertible states, the lower Fuchs--van de Graaf inequality is saturated only for trivial cases, while the upper Fuchs--van de Graaf inequality is saturated exactly when the geometric mean operator $M \coloneqq \rho^{-1} \# \sigma$ and the difference operator $\rho - \sigma$ are simultaneously diagonalizable, \emph{and} there exists some constant $c \in (0,1]$ such that the eigenvalues of $M$ are all either $c$ or $\frac{1}{c}$.

\begin{theorem}\label{QuantumFvdGTheorem}

With $\mathcal{S}_1$ and $\mathcal{S}_2$ as defined in Eqs.\ (\ref{S1Def}) and (\ref{S2Def}) respectively, we have
\begin{equation}\label{S1Characterization}
    \mathcal{S}_1 = \set{(\rho, \sigma) \in \mathcal{D}_{\textup{inv}}(\mathcal{H})^2 \mid \rho = \sigma}
\end{equation}
and
\begin{equation}\begin{aligned}\label{S2Characterization}
    \mathcal{S}_2 = \Big\{ &(\rho, \sigma) \in \mathcal{D}_{\textup{inv}}(\mathcal{H})^2 \Bigm| \rho = \sigma \lor \Big( \exists c \in (0,1), \\ &\spec(M) = \set{c, \frac{1}{c}} \land [M, \rho - \sigma] = 0 \Big) \Big\},
\end{aligned}\end{equation}
where
\begin{equation*}
    M \coloneqq \rho^{-1} \# \sigma = \M
\end{equation*}
is the operator geometric mean between $\rho^{-1}$ and $\sigma$, and $\spec(M)$ denotes the spectrum of $M$.

\end{theorem}

\begin{proof}
We first show one way to prove the left-hand inequality in Eq.~(\ref{QuantumFvdGInequalities}) for arbitrary $\rho, \sigma \in \mathcal{D}(\mathcal{H})$. Observe that for any $\Lambda_T \in \mathcal{T}(\rho, \sigma)$ and $\Lambda_F \in \mathcal{F}(\rho, \sigma)$, we have
\begin{equation}\begin{aligned}\label{QuantumFvdGInequality1}
    1 - \F{\rho}{\sigma} 
    &= 1 - \Fc{\tr{\Lambda_F \rho}}{\tr{\Lambda_F \sigma}} \\
    &\leq \Tc{\tr{\Lambda_F \rho}}{\tr{\Lambda_F \sigma}} \\
    &\leq \Tc{\tr{\Lambda_T \rho}}{\tr{\Lambda_T \sigma}} \\
    &= \T{\rho}{\sigma},
\end{aligned}\end{equation}
where the second line follows from Eq.~(\ref{ClassicalFvdGInequalities}) and the third line follows from Eq.~(\ref{TraceDistanceOptimization}). 

With this, we can now prove Eq.~(\ref{S1Characterization}). Restrict attention to invertible $\rho, \sigma \in \mathcal{D}_{\text{inv}}(\mathcal{H})$. First note that if $\rho = \sigma$, then clearly $1 - \F{\rho}{\sigma} = \T{\rho}{\sigma}$. Conversely, suppose that $1 - \F{\rho}{\sigma} = \T{\rho}{\sigma}$. Let $\set{\ket{e_x}}_{\xinX}$ be an orthonormal eigenbasis for the operator geometric mean $M$, and denote the corresponding projective measurement by $\Lambda_F$, in which case by Lemma \ref{VariationalFidelityLemma} we have
\begin{equation}\label{LambdaF}
    \Lambda_F \coloneqq \set{\ketbra{e_x}{e_x}}_{\xinX} \in \mathcal{F}(\rho, \sigma).
\end{equation}
Since $\Lambda_F \in \mathcal{F}(\rho, \sigma)$, the inequalities in the bound~(\ref{QuantumFvdGInequality1}) hold (given some arbitrary choice of $\Lambda_T \in \mathcal{T}(\rho, \sigma)$).
But since we assumed that $1 - \F{\rho}{\sigma} = \T{\rho}{\sigma}$, both inequalities in the bound must in fact be equalities. Focusing on the first inequality, we see that $\Lambda_F \in \mathcal{C}_1(\rho, \sigma)$. Then by Lemma \ref{ClassicalFvdGInequalitiesLemma}, we have
\begin{equation*}\begin{aligned}
    \forall \xinX, &\braket{e_x}{\rho \, e_x} = \braket{e_x}{\sigma \, e_x} \, \lor \\ &\braket{e_x}{\rho \, e_x} = 0 \lor \braket{e_x}{\sigma \, e_x} = 0.
\end{aligned}\end{equation*}
Since $\rho$ and $\sigma$ are assumed to be invertible and all $\ket{e_x}$ are nonzero, this implies
\begin{equation}\begin{aligned}\label{Star1}
    \forall \xinX, \braket{e_x}{\rho \, e_x} = \braket{e_x}{\sigma \, e_x}.
\end{aligned}\end{equation}
Now, for each $\xinX$, let $c_x > 0$ be the eigenvalue of $M$ corresponding to $\ket{e_x}$. Since $\sigma = M \rho M$, Eq.~(\ref{Star1}) implies
\begin{equation*}\begin{aligned}
    \forall \xinX, \braket{e_x}{\rho \, e_x} = c_x^2 \braket{e_x}{\rho \, e_x}.
\end{aligned}\end{equation*}
Since $\rho > 0$, this implies
\begin{equation*}\begin{aligned}
    \forall \xinX, c_x = 1.
\end{aligned}\end{equation*}
Thus, $M = \mathbbm{1}$, which implies that $\rho = \sigma$, as needed.

Next, we show one way to prove the right-hand inequality in Eq.~(\ref{QuantumFvdGInequalities}) for arbitrary $\rho, \sigma \in \mathcal{D}(\mathcal{H})$. Observe that for any $\Lambda_T \in \mathcal{T}(\rho, \sigma)$ and $\Lambda_F \in \mathcal{F}(\rho, \sigma)$, we have
\begin{equation}\begin{aligned}\label{QuantumFvdGInequality2}
    \T{\rho}{\sigma}
    &= \Tc{\tr{\Lambda_T \rho}}{\tr{\Lambda_T \sigma}} \\
    &\leq \sqrt{1-\Fc{\tr{\Lambda_T \rho}}{\tr{\Lambda_T \sigma}}^2} \\
    &\leq \sqrt{1-\Fc{\tr{\Lambda_F \rho}}{\tr{\Lambda_F \sigma}}^2} \\
    &=\sqrt{1-\F{\rho}{\sigma}^2},
\end{aligned}\end{equation}
where the second line follows from Eq.~(\ref{ClassicalFvdGInequalities}) and the third line follows from Eq.~(\ref{FidelityOptimization}). 

With this, we now prove Eq.~(\ref{S2Characterization}). Restrict attention to invertible $\rho, \sigma \in \mathcal{D}_{\text{inv}}(\mathcal{H})$. First note that if $\rho = \sigma$, then clearly $\T{\rho}{\sigma} = \sqrt{1-\F{\rho}{\sigma}^2}$, so this case is trivial. Otherwise, suppose that
\begin{equation*}
    \exists c \in (0,1), \spec(M) = \set{c, \frac{1}{c}} \land [M, \rho - \sigma] = 0.
\end{equation*}
Then there exists a basis $\set{\ket{e_x}}_{\xinX}$ for $\mathcal{H}$ that simultaneously diagonalizes $M$ and $\rho - \sigma$. Moreover, there exists $c \in (0,1)$ such that for all $\xinX$, $M \ket{e_x} = c \, \ket{e_x}$ or $M \ket{e_x} = \frac{1}{c} \, \ket{e_x}$. Define $\Lambda \coloneqq \set{\ketbra{e_x}{e_x}}_{\xinX} \in \mathcal{M}$. By Lemma \ref{VariationalTraceDistanceLemma}, $\Lambda \in \mathcal{T}(\rho, \sigma)$, and by Lemma \ref{VariationalFidelityLemma}, $\Lambda \in \mathcal{F}(\rho, \sigma)$. Thus, the inequalities in the bound in~(\ref{QuantumFvdGInequality2}) hold with $\Lambda_T \coloneqq \Lambda \eqqcolon \Lambda_F$. Now, since $\Lambda_T = \Lambda_F$, the third line in Eq.~(\ref{QuantumFvdGInequality2}) becomes an equality. Moreover, for each $\xinX$, we have
\begin{equation*}\begin{aligned}
    \braket{e_x}{\sigma \, e_x} &= \braket{e_x}{M \rho M \, e_x} \\
    &=  \begin{cases}
            c^2 \braket{e_x}{\rho \, e_x} & M \ket{e_x} = c \ket{e_x} \\
            \frac{1}{c^2} \braket{e_x}{\rho \, e_x} & M \ket{e_x} = \frac{1}{c} \ket{e_x} \\
        \end{cases}.
\end{aligned}\end{equation*}
With $b \coloneqq c^2 \in (0,1)$, we see that by Lemma \ref{ClassicalFvdGInequalitiesLemma}, $\Lambda \in \mathcal{C}_2(\rho, \sigma)$. Thus, the second line in Eq.~(\ref{QuantumFvdGInequality2}) becomes an equality. Then we have
\begin{equation*}
    \T{\rho}{\sigma} = \sqrt{1-\F{\rho}{\sigma}^2},
\end{equation*}
as needed.

Conversely, suppose that $\T{\rho}{\sigma} = \sqrt{1-\F{\rho}{\sigma}^2}$. 
Let $\set{\ket{f_y}}_{\yinX}$ be an orthonormal eigenbasis for $\rho - \sigma$, and denote the corresponding projective measurement by $\Lambda_T$, 
in which case by Lemma \ref{VariationalTraceDistanceLemma} we have
\begin{equation}\label{LambdaT}
    \Lambda_T \coloneqq \set{\ketbra{f_y}{f_y}}_{\yinX} \in  \mathcal{T}(\rho, \sigma).
\end{equation}
Since $\Lambda_T \in \mathcal{T}(\rho, \sigma)$, the inequalities in the bound~(\ref{QuantumFvdGInequality2}) hold (given some arbitrary choice of $\Lambda_F \in \mathcal{F}(\rho, \sigma)$).
But since we assumed that $\T{\rho}{\sigma} = \sqrt{1-\F{\rho}{\sigma}^2}$, both inequalities in the bound must in fact be equalities. 
Thus we have $\Lambda_T \in \mathcal{C}_2(\rho, \sigma)$ and $\Lambda_T \in \mathcal{F}(\rho, \sigma)$ (i.e.~$\Lambda_T$ must also be a fidelity-preserving measurement). Then by Lemma \ref{VariationalFidelityLemma}, $\set{\ket{f_y}}_{\yinX}$ is an orthonormal basis for $M$, so $[M, \rho - \sigma] = 0$. Moreover, by Lemma \ref{ClassicalFvdGInequalitiesLemma}, we have
\begin{equation*}\begin{aligned}
    &\parenth{\forall \yinX, \braket{f_y}{\rho \, f_y} = \braket{f_y}{\sigma \, f_y}} \, \lor \\
    &\parenth{\forall \yinX, \braket{f_y}{\rho \, f_y} = 0 \lor \braket{f_y}{\sigma \, f_y} = 0} \, \lor \\
    &\parenth{\begin{aligned} &\exists b \in (0,1), \forall \yinX, \\ &\braket{f_y}{\sigma \, f_y} = b \braket{f_y}{\rho \, f_y} \lor \braket{f_y}{\sigma \, f_y} = \frac{1}{b} \braket{f_y}{\rho \, f_y} \end{aligned}}.
\end{aligned}\end{equation*}
Now, if the first line above holds, then we appeal to a previous argument (see Eq.~(\ref{Star1})) and conclude that $\rho = \sigma$. If the second line above holds, then we obtain a contradiction, since $\rho$ and $\sigma$ are assumed to be invertible and all $\ket{f_y}$ are nonzero. Now, suppose that the third line above holds. For each $\yinX$, let $c_y > 0$ be the eigenvalue of $M$ corresponding to $\ket{f_y}$. With $\sigma = M \rho M$, we then have
\begin{equation*}\begin{aligned}
    &\exists b \in (0,1) , \forall \yinX, \\
    &c_y^2 \braket{f_y}{\rho \, f_y} = b \braket{f_y}{\rho \, f_y} \lor c_y^2 \braket{f_y}{\rho \, f_y} = \frac{1}{b} \braket{f_y}{\rho \, f_y}.
\end{aligned}\end{equation*}
With $c \coloneqq \sqrt{b} \in (0,1)$ and $\rho > 0$, this implies
\begin{equation*}\begin{aligned}
    \exists c \in (0,1), \forall \yinX, c_y = c \lor c_y = \frac{1}{c}.
\end{aligned}\end{equation*}
Thus, every eigenvalue $c_y$ of $M$ is either $c$ or $\frac{1}{c}$, so $\spec(M) = \set{c,\frac{1}{c}}$, as needed.
\end{proof}

\subsection{Noninvertible states}

How can we generalize Theorem \ref{QuantumFvdGTheorem} to arbitrary density operators $\rho, \sigma \in \mathcal{D}(\mathcal{H})$? To begin addressing this question, note that Lemmas \ref{VariationalTraceDistanceLemma} and \ref{ClassicalFvdGInequalitiesLemma} already apply to arbitrary $\rho, \sigma \in \mathcal{D}(\mathcal{H})$, but Lemma \ref{VariationalFidelityLemma} applies only to invertible $\rho, \sigma$. Next, note that Theorem \ref{QuantumFvdGTheorem} appears to generalize readily to noninvertible $\rho$ and $\sigma$ given a characterization of $\mathcal{F}(\rho, \sigma)$ for noninvertible $\rho$ and $\sigma$. Thus, it appears that the main difficulty in generalizing Theorem \ref{QuantumFvdGTheorem} is with generalizing Lemma \ref{VariationalFidelityLemma} to arbitrary $\rho$ and $\sigma$. In other words, if we can characterize the rank-1 projective measurements $\Lambda \in \mathcal{M}$ which achieve the optimum in Eq.~(\ref{FidelityOptimization}) for arbitrary $\rho$ and $\sigma$, i.e.\ if we can characterize $\mathcal{F}(\rho, \sigma)$ for arbitrary $\rho$ and $\sigma$, then it seems relatively straightforward to extend our characterization of the sets $\mathcal{S}_1,\mathcal{S}_2 \subseteq \mathcal{D}_{\text{inv}}(\mathcal{H})$ to a characterization of similarly defined sets $\tilde{\mathcal{S}}_1,\tilde{\mathcal{S}}_2 \subseteq \mathcal{D}(\mathcal{H})$ for arbitrary $\rho$ and $\sigma$.

Thus, we discuss possibilities for characterizing the set $\mathcal{F}(\rho, \sigma)$ for arbitrary $\rho$ and $\sigma$. We highlight that the proofs of the variational characterization~\eqref{FidelityOptimization} in e.g.~\cite{FC95,MW13} only yield specific examples of measurements attaining the optimum, rather than characterizing the set of all such measurements, which (as we shall soon discuss) appears significantly larger for some noninvertible states. As for the proof in~\cite{NC10}, it implicitly uses a compactness argument that also does not seem to yield a precise characterization of this set.

At first, it seems plausible that some appropriate generalization of the operator $M \coloneqq \rho^{-1} \# \sigma$ may be involved in the characterization of $\mathcal{F}(\rho, \sigma)$. However, it is not immediately clear what this generalization of $M$ might be. For example, one possibility for generalizing $M$ is to consider pseudoinverses. If $\rho \in \mathcal{D}(\mathcal{H})$ is noninvertible, the (Moore-Penrose) pseudoinverse is informally the ``inverse on the support" of $\rho$. That is (using $\rho^{-1}$ to denote the pseudoinverse of $\rho$), $\rho^{-1} \in \mathcal{L}(\mathcal{H})$ is an operator with the property that $\rho \, \rho^{-1} = \rho^{-1} \rho = \Pi_\rho$, where $\Pi_\rho$ is the projector onto the support of $\rho$. We could then try to define $M \coloneqq \rho^{-1} \# \sigma$ using pseudoinverses. However, with this approach, it is not possible to say that the set $\mathcal{F}(\rho, \sigma)$ for arbitrary $\rho, \sigma$ is just the same as in Lemma~\ref{VariationalFidelityLemma} except with $M$ defined via pseudoinverses. To see this, consider the case of pure states $\rho = \ketbra{\rho}{\rho}$ and $\sigma = \ketbra{\sigma}{\sigma}$. Then with this definition of $M$, we would have
\begin{equation*}
    \rho^{-1} \# \sigma = \abs{\braket{\rho}{\sigma}}\ketbra{\rho}{\rho},
\end{equation*}
but
\begin{equation*}
    \sigma^{-1} \# \rho = \abs{\braket{\rho}{\sigma}}\ketbra{\sigma}{\sigma}.
\end{equation*}
Now, observe that $\mathcal{F}$ is symmetric in its arguments, i.e.\ $\mathcal{F}(\rho, \sigma) = \mathcal{F}(\sigma, \rho)$, since the quantum and classical fidelities are both symmetric in their arguments. Thus, if Lemma~\ref{VariationalFidelityLemma} held for this choice of definition for $M$, then every $\Lambda \in \mathcal{F}(\rho, \sigma) = \mathcal{F}(\sigma, \rho)$ must be an eigenbasis of both $\rho^{-1} \# \sigma$ and $\sigma^{-1} \# \rho$. But when $\rho$ and $\sigma$ are pure states as above, this implies that $\ketbra{\rho}{\rho} \in \Lambda$ and $\ketbra{\sigma}{\sigma} \in \Lambda$, which is impossible whenever $\rho$ and $\sigma$ are distinct and nonorthogonal. This would imply that $\mathcal{F}(\rho, \sigma)$ is empty whenever $\rho$ and $\sigma$ are distinct, nonorthogonal pure states, which completely contradicts Eq.~(\ref{FidelityOptimization}). (Essentially, the fundamental issue here is that defining $M$ using pseudoinverses causes it to lose a symmetry property $\rho^{-1} \# \sigma = (\sigma^{-1} \# \rho)^{-1}$ that held for invertible operators.)

Another potential approach for generalizing $M$ is to note that for noninvertible operators, one can choose to define the operator geometric mean as (see e.g.~\cite{Bha06} page 211)
\begin{equation}\label{eq:geomean2}
    A \# B \coloneqq \lim_{\delta \to 0^+} (A+\delta \mathbbm{1}) \# (B+\delta \mathbbm{1}),
\end{equation}
where the right-hand-side can be computed using the definition~\eqref{eq:geomean} since $A+\delta \mathbbm{1}$ and $B+\delta \mathbbm{1}$ are both invertible for $\delta>0$. It can be shown that the limit in~\eqref{eq:geomean2} indeed exists for all positive semidefinite $A$ and $B$, although this definition of $A \# B$ has the drawback that it is not continuous with respect to $A$ and $B$~\cite{Bha06}.
However, this approach in our context faces the difficulty that the term in our result is not $\rho \# \sigma$, but rather $\rho^{-1} \# \sigma$. Therefore, even if we were to choose some generalized definition of the operator geometric mean for noninvertible operators, it would still not be enough by itself to resolve the issue of generalizing $\rho^{-1} \# \sigma$, since $\rho^{-1}$ is already ill-defined if $\rho$ is noninvertible.

Drawing on the above idea, however, we could still consider ``$\delta$-perturbed'' versions of $\rho$ and $\sigma$ (such that the perturbed versions are invertible), and analyze the $\delta\to 0^+$ limit in the broader context of our desired result rather than the operator geometric mean specifically. 
We sketch the starting points of this approach here, deferring further analysis to Appendix~\ref{app:perturb}.
To begin, consider the following states, where $d \coloneqq \dim \mathcal{H}$ (here we shall define the perturbations slightly differently from~\eqref{eq:geomean2}, in order to ensure that $\rho_\delta,\sigma_\delta$ are normalized states):
\begin{equation}\label{eq:perturb}
    \rho_\delta \coloneqq (1-\delta) \rho + \delta \frac{\mathbbm{1}}{d}, \quad
    \sigma_\delta \coloneqq (1-\delta) \sigma + \delta \frac{\mathbbm{1}}{d}.
\end{equation}
Using the (reverse) triangle inequality for angular distance, these states can be seen to satisfy
\begin{gather*}
    \left| \A{\rho_\delta}{\sigma_\delta} - \A{\rho}{\sigma}\right| \leq \negl(\delta),
\end{gather*}
where for brevity we use the notation $\negl(\delta)$ to indicate any expression such that $\lim_{\delta\to 0^+}\negl(\delta) = 0$. In other words, the perturbations as defined in~\eqref{eq:perturb} only change the angular distance (and thus also the fidelity) by an amount that vanishes in the $\delta\to 0^+$ limit. 

From this, we see that if for instance $\rho$ and $\sigma$ saturate the upper Fuchs--van de Graaf inequality, then $\rho_\delta$ and $\sigma_\delta$ ``approximately saturate'' it as well, in the sense that
\begin{equation}\label{eq:appFvdG}
\begin{gathered}
\T{\rho}{\sigma} - \sqrt{1 - \F{\rho}{\sigma}^2} = 0\\
\implies \left| \T{\rho}{\sigma} - \sqrt{1 - \F{\rho_\delta}{\sigma_\delta}^2} \right| \leq \negl(\delta).
\end{gathered}
\end{equation}
(In the above, we have only considered perturbing the fidelity term rather than the trace-distance term; the reason for this will become apparent in our more detailed analysis in Appendix~\ref{app:perturb}.)
Following this form of reasoning, we could continue onwards and attempt to repeat the proof of Theorem~\ref{QuantumFvdGTheorem}, except with ``approximate equalities'' instead of equalities. We were able to make some progress with this approach, which we describe in Appendix~\ref{app:perturb}. However, it still does not seem sufficient to resolve the question of extending Theorem~\ref{QuantumFvdGTheorem} to arbitrary noninvertible states, and perhaps raises the question of whether some notion of the geometric mean operator is even the ``right'' object to consider in this characterization.
We also remark that the above considerations seem to suggest the main challenges for noninvertible states mostly arise when $\rho$ is noninvertible --- at a high level, it seems that the proofs in the preceding sections should basically carry through for noninvertible $\sigma$ as long as $\rho$ is still invertible.

To end off, we highlight what seems to be a significant broad obstacle in generalizing our characterization of $\mathcal{F}(\rho, \sigma)$ to arbitrary $\rho$ and $\sigma$, by presenting the characterization of this set for the case of pure states $\rho = \ketbra{\rho}{\rho}$ and $\sigma = \ketbra{\sigma}{\sigma}$. While it is possible to work through the proof of Lemma \ref{VariationalFidelityLemma} to study this special case, an easier approach is via direct computation. For any $\Lambda \coloneqq \set{\ketbra{e_x}{e_x}}_{\xinX} \in \mathcal{M}$, we have
\begin{equation*}\begin{aligned}
    \F{\rho}{\sigma} &= \abs{\braket{\rho}{\sigma}} \\
    &= \abs{\sum_{\xinX} \braket{\rho}{e_x} \braket{e_x}{\sigma}} \\
    &\leq \sum_{\xinX} \abs{\braket{\rho}{e_x} \braket{e_x}{\sigma}} \\
    &= \Fc{\tr{\Lambda \rho}}{\tr{\Lambda \sigma}}
\end{aligned}\end{equation*}
Now, equality holds in the above iff the triangle inequality in the third line is a strict equality. This occurs iff there exists $\theta \in \mathbb{R}$ such that for every $\xinX$, $\arg \braket{\rho}{e_x} \braket{e_x}{\sigma} \equiv \theta \!\! \mod 2\pi \lor \braket{\rho}{e_x} \braket{e_x}{\sigma} = 0$. Thus, for pure $\rho$ and $\sigma$, we have
\begin{equation*}\begin{aligned}
    \mathcal{F}(\ketbra{\rho}{\rho},\ketbra{\sigma}{\sigma}) = \Big\{&\Lambda \in \mathcal{M} \mid \exists \theta \in \mathbb{R}, \forall \xinX, \\
    &\arg\braket{e_x}{\sigma} - \arg\braket{e_x}{\rho} \equiv \theta \!\!\!\! \mod 2\pi \, \lor \\ &\braket{e_x}{\sigma} = 0 \lor \braket{e_x}{\rho} = 0 \Big\}.
\end{aligned}\end{equation*}
Note that the global phases of the representative state vectors $\ket{\rho},\ket{\sigma}$ can be chosen arbitrarily, as the description of the above set is invariant under such changes of phase.

With this, we see that $\mathcal{F}(\rho, \sigma)$ is a much larger set in the pure-state case than in the invertible case: in the latter case, all $\Lambda \in \mathcal{F}(\rho, \sigma)$ are essentially equivalent up to degeneracy in the spectral decomposition of $M$, whereas in the former case, we have for instance that \emph{any} orthonormal basis in which all the components $\braket{e_x}{\rho},\braket{e_x}{\sigma}$ are real and non-negative yields a measurement in $\mathcal{F}(\rho, \sigma)$ (note that it is easy to construct examples of pure states $\ket{\rho}$ and $\ket{\sigma}$ such that many orthonormal bases do have this property). Thus, it seems unclear how to generalize our characterization of $\mathcal{F}(\rho, \sigma)$ to noninvertible states, since any such generalization must capture the special case above for pure $\rho$ and $\sigma$.

\section{Conclusion}\label{sec:conclusion}

In this work, we derived a continuity bound for the conditional entropy of quantum-classical states with respect to angular distance. This bound satisfies both of conditions \ref{DesiredBoundConditionIndependence} and \ref{DesiredBoundConditionDifferentiable}, which are desirable in many applications to quantum key distribution \cite{SBV+21, Upadhyaya_2021}. However, in those applications, a continuity bound for classical-quantum states is usually required. Further work is thus needed to extend our result to classical-quantum states. Numerical evidence suggests that our bound in Eq.~(\ref{ContinuityBound}) may indeed hold for such states, and with minor modifications, our bound may also be valid for fully quantum states.

To find such a generalization, one approach could be to consider the large body of work on entropic continuity bounds in terms of trace distance, for instance~\cite{Win16,arx_HD19,arx_JD20,arx_MD22,arx_BCG+22}, and study whether any of the proof approaches in those works could be modified to use angular distance instead.
To begin, the works~\cite{arx_HD19,arx_JD20} used techniques from majorization theory to prove entropic continuity bounds, including for families of R\'{e}nyi entropies. However, those techniques do not seem straightforward to apply when the conditioning system is quantum. For such scenarios, continuity bounds were derived in~\cite{arx_MD22} for R\'{e}nyi entropies, and a recent work~\cite{arx_BCG+22} proved an ``almost locally affine'' property of the relative entropy that (amongst other results) reproduces the bound~\cite{Win16}. However, qualitatively speaking, the approaches in those works seem to rely on studying ``additive perturbations'' to the states $\rho,\sigma$, which are naturally related to trace distance but seem more difficult to express in terms of angular distance. 

We also note the observation in~\cite{Wil20} that continuity bounds for fully classical conditional entropies can often be quite ``generically'' extended to the quantum-classical case (essentially, whenever the distance measure satisfies a data-processing inequality), but such a generic extension for the classical-quantum or fully quantum cases seems to require new ideas or techniques. If this is indeed so, extending our result to cover those latter cases would require exploiting some specific property of angular distance that is not shared by trace distance, for instance the characterization via Uhlmann's theorem.

In the second part of our work, to relate previous continuity bounds based on trace distance to our result based on angular distance, we explored the relationship between trace distance and angular distance via the Fuchs--van de Graaf inequalities. In particular, we derived necessary and sufficient conditions for invertible states to saturate either side of the Fuchs--van de Graaf inequalities relating the trace distance and fidelity. We remark that this may have independent applications in other topics, such as computing keyrates for QKD; we briefly outline this in appendix~\ref{app:keyrates}.

In addition, we showed that generalizing our result to noninvertible states appears nontrivial, and that this generalization is closely related to characterizing the set of rank-1 projective measurements which preserve the fidelity. Future work could continue by generalizing Lemma \ref{VariationalFidelityLemma} and Theorem \ref{QuantumFvdGTheorem} to noninvertible states; it appears that generalizing Lemma \ref{VariationalFidelityLemma} is the difficult part, whereas generalizing Theorem \ref{QuantumFvdGTheorem} after that seems to be relatively straightforward.

We note that an alternative approach for such an analysis could be to utilize a proof of the upper Fuchs--van de Graaf inequality via Uhlmann's theorem, as presented in e.g.~\cite{NC10}. To give a high-level overview, that approach would yield (via the discussion at the beginning of Sec.\ \ref{sec:InvertibleStates}) that the set of states saturating the upper Fuchs--van de Graaf inequality is precisely the set of states such that the Uhlmann purifications have the same trace distance as the original states. However, the construction of the Uhlmann purification involves a polar decomposition that appears similar to the one which arises in the proof of Lemma \ref{VariationalFidelityLemma}, and it seems unclear whether it can be analyzed more fruitfully.

\begin{acknowledgments}

We thank Norbert L\"{u}tkenhaus, Joseph M.~Renes and Jinzhao Wang
for helpful feedback and discussions. 
Financial support for this work has been provided by the Natural Sciences and Engineering Research Council of Canada (NSERC) Alliance, and Huawei Technologies Canada Co., Ltd.
Numerical data was generated using Mathematica\textsuperscript{\textregistered} 12.1.

\end{acknowledgments}

\appendix

\section{Conversions between continuity bounds}
\label{app:compare}

Using the Fuchs--van de Graaf inequalities in Eq.\ (\ref{QuantumFvdGInequalities}), we can convert continuity bounds in trace distance to continuity bounds in angular distance and vice versa. For states $\rho, \sigma \in \mathcal{D}(\mathcal{H})$ with trace distance $T \coloneqq \T{\rho}{\sigma}$ and angular distance $A \coloneqq \A{\rho}{\sigma}$, the upper Fuchs--van de Graaf inequality gives $T \leq \sin A$. However, as illustrated in the calculations at the beginning of Sec.~\ref{sec:FvdG}, combining this with the previously known bounds (\ref{AudenaertBound}) and (\ref{WinterBound}) in trace distance yields results which scale badly with $A$ as compared to our result.

In the reverse direction, the lower Fuchs--van de Graaf inequality gives $A \leq \arccos(1-T)$. Plugging this into our continuity bound (\ref{ContinuityBound}) yields an upper bound of $u(d_A) \arccos(1-T)$ on the difference in conditional entropies. At small $T$, we can approximate this as
\begin{equation*}
u(d_A) \arccos(1-T) \approx u(d_A) \sqrt{2T}.
\end{equation*}

For comparison, in the bound~\eqref{AudenaertBound} for the unconditioned entropies (which corresponds to the $d_B=1$ case in our bound), we can use an approximation $\hbinary{x} \lesssim 2\sqrt{x}$ at small $x$ for the binary entropy function to obtain
\begin{align*}
T \ln\parenth{d_A-1} + \hbinary{T} 
&\lesssim T \ln\parenth{d_A-1} + 2\sqrt{T} \\
&< (\ln\parenth{d_A-1} + 2)\sqrt{T}.
\end{align*}
From the first line, we see that at small $T$, both the bound obtained via our result and the bound from~\eqref{AudenaertBound} basically scale on the order of $\sqrt{T}$; however, the coefficient in front of that term is quite different. The second line above reveals that our bound cannot outperform the bound~\eqref{AudenaertBound} (as should be expected, since the latter is tight): for $d_A \leq 3$ one can verify numerically that $\ln\parenth{d_A-1} + 2 \leq u(d_A) $, while for $d_A \geq 4$ we have $\ln\parenth{d_A-1} + 2 \leq \ln\parenth{d_A-1} + \ln\parenth{d_A} \leq 2\ln\parenth{d_A} \leq u(d_A) $. 

As for the bound~\eqref{WinterBound} for conditional entropies (which is nearly tight but not exactly so~\cite{Win16}), we first highlight that for quantum-classical states, the prefactor of $2$ on the first term in that bound can be omitted. Since our bound only holds for quantum-classical states, we can compare it to that version. However, with that change, the modified version of \eqref{WinterBound} is quite close to the unconditioned-entropy bound~\eqref{AudenaertBound} at small $T$, and hence (for small $T$, at least) our result also cannot provide an improvement over~\eqref{WinterBound} via a conversion of this form. 

\section{Challenges in purification-based arguments}
\label{app:purification}

At first sight, it might appear that since we are considering angular distance as our metric, it would be useful to consider purifications, since we could for instance apply results such as Uhlmann's theorem. Unfortunately, this appears to encounter difficulties regarding the conditional entropies, as we shall now describe.

We first observe that for fully quantum states, the following issue arises. Take any $\rho, \sigma \in \mathcal{D}\parenth{\mathcal{H}^A \otimes \mathcal{H}^B}$ such that $\rho_A = \sigma_A$. If we purify $\rho$ and $\sigma$ to some pure states $\rho_{ABR}$ and $\sigma_{ABR}$, then we have $\operatorname{H} \parenth{A|BR}_{\rho} = -\operatorname{H} \parenth{A}_{\rho} = -\operatorname{H} \parenth{A}_{\sigma} = \operatorname{H} \parenth{A|BR}_{\sigma}$, i.e.\ the difference in conditional entropies $\operatorname{H}\parenth{A|BR}$ is zero. Hence these conditional entropies cannot give us any information about the difference in the original conditional entropies $\operatorname{H} \parenth{A|B}$. If we instead consider the entropies $\operatorname{H}\parenth{AR|B}$, an analogous problem arises whenever $\rho_B = \sigma_B$.

For classical-quantum states, we could consider a modified version of this approach, by taking ``individual purifications'' of the conditional quantum states on the $\mathcal{H}^B$ systems. Specifically, for a state $\rho_{AB}$ of the form $\rho_{AB} = \sum_a p_a \ketbra{a}{a} \otimes \rho^{(a)}_{B}$, we can consider an extension $\rho_{ABR} = \sum_a  p_a \ketbra{a}{a} \otimes \rho^{(a)}_{BR}$ such that each $\rho^{(a)}_{BR}$ is a purification of $\rho^{(a)}_{B}$. For simplicity, let us suppose here that the other state $\sigma_{AB}$ satisfies $\sigma_A=\rho_A$ (and is also classical-quantum).
In that case, by an appropriate application of Uhlmann's theorem, it can be shown that for any extension of $\rho_{AB}$ as described above, we can construct an analogous extension of $\sigma_{AB}$ with the property $\F{\rho_{AB}}{\sigma_{AB}} = \F{\rho_{ABR}}{\sigma_{ABR}}$. This seems promising as it preserves the angular distance --- if we could furthermore show that the difference in conditional entropies is nondecreasing under some such extension, this would imply that it suffices to consider pure conditional states in this scenario, simplifying the analysis. (In fact, if that were true, it would already be sufficient to yield a continuity bound with the desired properties by noting that the pure conditional states would all be supported on a subspace of dimension at most $2d_A$, although the resulting Lipschitz constant might be suboptimal.)

However, this encounters the following obstacle: 
one can show~\footnote{To outline the key ideas, in~\cite{RFZ10} the following bound was derived (see Eq.~(21) of that work): $\HCond{\rho} \geq 1-\hbinary{\left(1-\F{\rho^{(0)}_{B}}{\rho^{(1)}_{B}}\right)/2}$ for $\rho_{AB}$ of the form $\rho_{AB} = \sum_{a\in\{0,1\}} (1/2) \ketbra{a}{a} \otimes \rho^{(a)}_{B}$, with equality holding when the two conditional states $\rho^{(a)}_{B}$ are pure.
Take any $\rho_{AB}$ such that this bound is a strict inequality, then observe that any purifications of the conditional states must satisfy $\F{\rho^{(0)}_{B}}{\rho^{(1)}_{B}} \geq \F{\rho^{(0)}_{BR}}{\rho^{(1)}_{BR}}$, and use the fact that the bound from~\cite{RFZ10} becomes an equality for pure $\rho^{(a)}_{BR}$.} that there exist states $\rho_{AB}$ such that for some $\delta>0$, \emph{any} extension $\rho_{ABR}$ in the above sense satisfies $ \operatorname{H} \parenth{A|BR}_{\rho} < \HCond{\rho}  - \delta$; i.e.\ the conditional entropy of the extension is bounded away from the original value by a constant.
In that case, if we take any classical-quantum $\sigma_{AB}$ with $\HCond{\sigma} = 0$ (in which case any extension in the above sense must have $\operatorname{H} \parenth{A|BR}_{\sigma} = 0$ as well, due to strong subadditivity), we have 
\begin{align*}
\Big| \! \operatorname{H} \parenth{A|BR}_{\rho} - \operatorname{H} \parenth{A|BR}_{\sigma} \! \Big| 
=& \operatorname{H} \parenth{A|BR}_{\rho} \\
<& \HCond{\rho} - \delta \\
=& \Big| \! \HCond{\rho} - \HCond{\sigma} \! \Big| - \delta,
\end{align*}
or in other words, the difference in entropies $\operatorname{H}\parenth{A|BR}$ of the extensions is strictly smaller than the original difference (more precisely, bounded away from it by a constant amount). This is problematic for the goal of using these extensions to bound the original entropy difference. While this obstacle may not be impossible to overcome, it does suggest that there may not be a straightforward proof using this approach.

\section{Proofs of Lemmas}
\label{app:proofsFvdG}

\subsection{Proof of Lemma \ref{VariationalTraceDistanceLemma}}\label{ProofVariationalTraceDistanceLemma}

\begin{proof}
Let $\Lambda \coloneqq \set{\ketbra{e_x}{e_x}}_{\xinX} \in \mathcal{M}$ be arbitrary. We first prove that \begin{equation*}
    \Tc{\tr{\Lambda \rho}}{\tr{\Lambda \sigma}} \leq \T{\rho}{\sigma}.
\end{equation*}
To see this, consider a spectral decomposition
\begin{equation*}
    \rho - \sigma = \sum_{\yinX} \lambda_y \ketbra{f_y}{f_y}
\end{equation*}
of $\rho - \sigma$, where the $\lambda_y \in \mathbb{R}$ and the $\ket{f_y}$ form an orthonormal basis for $\mathcal{H}$. Define the positive and negative parts of $\rho - \sigma$ as
\begin{equation*}\begin{aligned}
    P &\coloneqq \sum_{\yinX, \lambda_y > 0} \lambda_y \ketbra{f_y}{f_y} \\
    Q &\coloneqq \sum_{\yinX, \lambda_y < 0} \abs{\lambda_y} \ketbra{f_y}{f_y},
\end{aligned}\end{equation*}
respectively. Note that $\rho - \sigma = P - Q$ and $\abs{\rho - \sigma} \coloneqq \sqrt{(\rho - \sigma)^\dagger (\rho - \sigma)} = P + Q$. Then for any $\xinX$, we have
\begin{equation}\begin{aligned}\label{TraceDistancexEstimate}
    \abs{\braket{e_x}{\parenth{\rho - \sigma} e_x}}
    &= \abs{\braket{e_x}{\parenth{P - Q} e_x}} \\
    &\leq \braket{e_x}{\parenth{P + Q} e_x} \\
    &= \braket{e_x}{\abs{\rho - \sigma} e_x},
\end{aligned}\end{equation}
where in the second line we used that $P, Q \geq 0$. Thus, we have
\begin{equation*}\begin{aligned}
    \Tc{\tr{\Lambda \rho}}{\tr{\Lambda \sigma}}
    &= \frac{1}{2} \sum_{\xinX} \abs{\braket{e_x}{\rho \, e_x} - \braket{e_x}{\sigma \, e_x}} \\
    &\leq \frac{1}{2} \sum_{\xinX} \braket{e_x}{\abs{\rho - \sigma}e_x} \\
    &= \frac{1}{2} \tr{\abs{\rho - \sigma}} \\
    &=\T{\rho}{\sigma},
\end{aligned}\end{equation*}
as needed. Now, in the second line above, we applied Eq.~(\ref{TraceDistancexEstimate}) for each $\xinX$. Thus, equality holds in the above iff
\begin{equation*}
    \forall \xinX, \abs{\braket{e_x}{(P - Q) e_x}} = \braket{e_x}{(P + Q) e_x}.
\end{equation*}
Since $P, Q \geq 0$, this occurs iff
\begin{equation*}
    \forall \xinX, \braket{e_x}{P e_x} = 0 \lor \braket{e_x}{Q e_x} = 0.
\end{equation*}
Since any positive semidefinite $A \in \mathcal{L}(\mathcal{H})$ satisfies $\braket{e_x}{A e_x} = 0 \iff \ket{e_x} \in \ker A$, the above is equivalent to
\begin{equation*}
    \forall \xinX, \ket{e_x} \in \ker P \lor \ket{e_x} \in \ker Q,
\end{equation*}
as needed. To see that $\mathcal{T}(\rho, \sigma)$ is nonempty, consider $\tilde{\Lambda} \coloneqq \set{\ketbra{f_y}{f_y}}_{\yinX} \in \mathcal{M}$. Clearly, this satisfies
\begin{equation*}
    \forall \yinX, \ket{f_y} \in \ker P \lor \ket{f_y} \in \ker Q,
\end{equation*}
so $\tilde{\Lambda}$ is an optimizing projective measurement, as needed.
\end{proof}

\subsection{Proof of Lemma \ref{VariationalFidelityLemma}}\label{ProofVariationalFidelityLemma}

\begin{proof}
Let $\Lambda \coloneqq \set{\ketbra{e_x}{e_x}}_{\xinX} \in \mathcal{M}$ be arbitrary. We first prove that \begin{equation*}
    \F{\rho}{\sigma} \leq \Fc{\tr{\Lambda \rho}}{\tr{\Lambda \sigma}}.
\end{equation*}
To see this, note that by the polar decomposition, there exists a unitary $U$ on $\mathcal{H}$ such that
\begin{equation}\label{PolarDecomposition}
    \K = \KU.
\end{equation}
Then we have
\begin{equation*}\begin{aligned}
    \F{\rho}{\sigma}
    &= \norm{\srss}_1 \\
    &= \tr{\K} \\
    &= \tr{\KU} \\
    &= \sum_{\xinX} \tr{\KxU} \\
    &= \abs{\sum_{\xinX} \tr{\KxU}},
\end{aligned}\end{equation*}
where the last line follows since $\F{\rho}{\sigma} \geq 0$. Applying the triangle inequality followed by the Cauchy--Schwarz inequality, we continue with
\begin{equation}\begin{aligned} \label{FidelityInequality}
    \F{\rho}{\sigma}
    &= \abs{\sum_{\xinX} \tr{\KxU}} \\
    &\leq \sum_{\xinX} \abs{\tr{\sqrt{\rho} \, \ketbra{e_x}{e_x} \ketbra{e_x}{e_x} \sqrt{\sigma} \, U}} \\
    &\leq \sum_{\xinX} \sqrt{\tr{\ketbra{e_x}{e_x} \, \rho}} \sqrt{\tr{\ketbra{e_x}{e_x} \, \sigma}} \\
    &= \Fc{\tr{\Lambda \rho}}{\tr{\Lambda \sigma}},
\end{aligned}\end{equation}
as needed. Now, equality occurs in the above iff the triangle inequality and the Cauchy--Schwarz inequality are both saturated. This occurs iff each term $\tr{\KxU}$ either is of a fixed complex phase or is equal to 0, and the set $\set{\ketbra{e_x}{e_x} \sqrt{\rho}, \ketbra{e_x}{e_x} \sqrt{\sigma} \, U} \subseteq \mathcal{L}\parenth{\mathcal{H}}$ is linearly dependent for all $\xinX$. To continue, observe that $\sum_{\xinX} \tr{\KxU} = \F{\rho}{\sigma} \geq 0$. In addition, observe that $\ketbra{e_x}{e_x} \sqrt{\rho} \neq 0$ and $\ketbra{e_x}{e_x} \sqrt{\sigma} \, U \neq 0$, since $\rho$ and $\sigma$ are assumed invertible and $\ket{e_x} \neq 0$. This implies that equality in Eq.~(\ref{FidelityInequality}) occurs iff the terms $\tr{\KxU}$ are all nonnegative and $\ketbra{e_x}{e_x} \sqrt{\rho} = c_x \ketbra{e_x}{e_x}  \sqrt{\sigma} \, U$ for all $\xinX$, where $c_x \in \mathbb{C} \setminus \set{0}$. Equivalently, equality holds iff
\begin{equation*}\begin{aligned}
    \forall \xinX, \, &\exists c_x \in \mathbb{C} \setminus \set{0}, \\
    &\braket{e_x}{\sqrt{\sigma} \, U \sqrt{\rho} \, e_x} \geq 0 \, \land \\
    &\ketbra{e_x}{e_x} \sqrt{\rho} = c_x \ketbra{e_x}{e_x} \sqrt{\sigma} \, U.
\end{aligned}\end{equation*}
Since $\rho$ is assumed invertible, this is equivalent to
\begin{equation*}\begin{aligned}
    \forall \xinX, \, &\exists c_x \in \mathbb{C} \setminus \set{0}, \\
    &\braket{e_x}{\power{\rho}{-} \KU \power{\rho}{-} \rho \, e_x} \geq 0 \, \land \\
    &\bra{e_x} = c_x \bra{e_x} \power{\rho}{-} \KU \power{\rho}{-}.
\end{aligned}\end{equation*}
Recalling Eq.~(\ref{PolarDecomposition}) and the definition of the operator geometric mean $M \coloneqq \rho^{-1} \# \sigma = \M$, this is equivalent to
\begin{equation*}\begin{aligned}
    \forall \xinX, \, &\exists c_x \in \mathbb{C} \setminus \set{0}, \\ &\braket{e_x}{M \rho \, e_x} \geq 0 \land \bra{e_x} M = \frac{1}{c_x} \bra{e_x}.
\end{aligned}\end{equation*}
Since $M = M^\dagger$, this is equivalent to
\begin{equation*}\begin{aligned}
    \forall \xinX, \, &\exists c_x \in \mathbb{C} \setminus \set{0}, \\ &\overline{c_x} \braket{e_x}{M \rho M e_x} \geq 0 \land M \ket{e_x} = \frac{1}{\overline{c_x}} \ket{e_x}.
\end{aligned}\end{equation*}
But since $M \rho M = \sigma$ and $\sigma > 0$ (since $\sigma$ is assumed invertible), this is equivalent to
\begin{equation*}
    \forall \xinX, \exists c_x > 0, M \ket{e_x} = \frac{1}{c_x} \ket{e_x}.
\end{equation*}
Thus, since $M > 0$, so the eigenvalues of $M$ are all positive, we see that equality in Eq.~(\ref{FidelityInequality}) holds iff $\set{\ket{e_x}}_{\xinX}$ is an eigenbasis for $M$, as needed. Clearly, an eigenbasis for $M$ exists, so $\mathcal{F}(\rho, \sigma)$ is nonempty, as needed.
\end{proof}

\subsection{Proof of Lemma \ref{ClassicalFvdGInequalitiesLemma}}\label{ProofClassicalFvdGInequalitiesLemma}

\begin{proof}
For notational simplicity, we work with arbitrary probability distributions $p$ and $q$ on a finite alphabet $\mathcal{X}$. We first prove the left-hand inequality in Eq.~(\ref{ClassicalFvdGInequalities}). Observe that
\begin{equation*}\begin{aligned}
    1 - \Fc{p}{q}
    &= \frac{1}{2} \sum_{\xinX} \parenth{p(x) + q(x) - 2 \sqrt{p(x)q(x)}} \\
    &= \frac{1}{2} \sum_{\xinX} \abs{\sqrt{p(x)} - \sqrt{q(x)}}^2 \\
    &\leq \frac{1}{2} \sum_{\xinX} \abs{\sqrt{p(x)} - \sqrt{q(x)}} \abs{\sqrt{p(x)} + \sqrt{q(x)}} \\
    &= \frac{1}{2} \sum_{\xinX} \abs{p(x) - q(x)} \\
    &= \Tc{p}{q},
\end{aligned}\end{equation*}
as needed. Now, equality holds in the above iff
\begin{equation*}\begin{aligned}
    \forall \xinX, &\abs{\sqrt{p(x)} - \sqrt{q(x)}} \abs{\sqrt{p(x)} - \sqrt{q(x)}} \\
    = &\abs{\sqrt{p(x)} - \sqrt{q(x)}} \abs{\sqrt{p(x)} + \sqrt{q(x)}},
\end{aligned}\end{equation*}
which occurs iff
\begin{equation*}\begin{aligned}
    \forall \xinX, &\abs{\sqrt{p(x)} - \sqrt{q(x)}} = 0 \\
    \lor &\abs{\sqrt{p(x)} - \sqrt{q(x)}} = \abs{\sqrt{p(x)} + \sqrt{q(x)}},
\end{aligned}\end{equation*}
which occurs iff
\begin{equation*}
    \forall \xinX, p(x) = q(x) \lor p(x) = 0 \lor q(x) = 0,
\end{equation*}
as needed. Next, we prove the right-hand inequality in Eq.~(\ref{ClassicalFvdGInequalities}). Observe that
\begin{equation*}\begin{aligned}
    \operatorname{T}_c&\parenth{p,q}^2 \\
    &= \parenth{\frac{1}{2} \sum_{\xinX} \abs{p(x) - q(x)}}^2 \\
    &= \frac{1}{4} \parenth{\sum_{\xinX} \abs{\sqrt{p(x)} - \sqrt{q(x)}} \abs{\sqrt{p(x)} + \sqrt{q(x)}}}^2 \\
    &\leq \frac{1}{4} \sum_{\xinX} \parenth{\sqrt{p(x)} - \sqrt{q(x)}}^2 \sum_{\xinX} \parenth{\sqrt{p(x)} + \sqrt{q(x)}}^2 \\
    &= \frac{1}{4} (2 - 2 \Fc{p}{q}))(2 + 2 \Fc{p}{q}) \\
    &= 1-\Fc{p}{q}^2,
\end{aligned}\end{equation*}
as needed. Note that in the third line above, we applied the Cauchy--Schwarz inequality with
\begin{equation*}\begin{aligned}
    u &\coloneqq \parenth{\abs{\sqrt{p(x)} - \sqrt{q(x)}}}_{\xinX} \in \mathbb{R}^d \\
    v &\coloneqq \parenth{\abs{\sqrt{p(x)} + \sqrt{q(x)}}}_{\xinX} \in \mathbb{R}^d.
\end{aligned}\end{equation*}
Thus, equality holds in the above iff $u$ and $v$ saturate the Cauchy--Schwarz inequality, which occurs iff $\set{u,v} \subseteq \mathbb{R}^d$ is linearly dependent. Since $v \neq 0$, this is equivalent to $u \in \Span \set{v}$, i.e.\ $\exists a \in \mathbb{R}, u = a \, v$. Since $u \geq 0$ and $v > 0$, this is equivalent to $\exists a \geq 0, u = a \, v$, i.e.\
\begin{equation*}\begin{aligned}
    \exists a \geq 0, \, &\forall \xinX, \\
    &\abs{\sqrt{p(x)} - \sqrt{q(x)}} = a \abs{\sqrt{p(x)} + \sqrt{q(x)}}.
\end{aligned}\end{equation*}
Breaking into cases, this is equivalent to
\begin{equation*}\begin{aligned}
    \exists a \geq 0, \, &\forall \xinX, \\
    &\sqrt{p(x)} - \sqrt{q(x)} = a \parenth{\sqrt{p(x)} + \sqrt{q(x)}} \lor \\ &\sqrt{q(x)} - \sqrt{p(x)} = a \parenth{\sqrt{p(x)} + \sqrt{q(x)}}.
\end{aligned}\end{equation*}
Rearranging, this is equivalent to
\begin{equation*}\begin{aligned}
    \exists a \geq 0, \, &\forall \xinX,\\
    &\frac{1-a}{1+a} \sqrt{p(x)} = \sqrt{q(x)} \lor \frac{1+a}{1-a} \sqrt{p(x)} = \sqrt{q(x)}.
\end{aligned}\end{equation*}
Now, since $\sqrt{p(x)},\sqrt{q(x)} \geq 0$ for all $\xinX$, we can restrict $a \in [0,1]$, so the above is equivalent to
\begin{equation*}\begin{aligned}
    \exists a \in [0,1], \, &\forall \xinX,\\
    &\frac{1-a}{1+a} \sqrt{p(x)} = \sqrt{q(x)} \lor \frac{1+a}{1-a} \sqrt{p(x)} = \sqrt{q(x)}.
\end{aligned}\end{equation*}
Treating $a=0$ and $a=1$ separately, this is equivalent to
\begin{equation*}\begin{aligned}
    &\parenth{\forall \xinX, p(x) = q(x)} \lor \parenth{\forall \xinX, p(x) = 0 \lor q(x) = 0} \lor \\
    &\parenth{\begin{aligned} &\exists a \in (0,1), \forall \xinX, \\ &q(x) = \parenth{\frac{1-a}{1+a}}^2 p(x) \lor q(x) = \parenth{\frac{1+a}{1-a}}^2 p(x) \end{aligned}}.
\end{aligned}\end{equation*}
Since $p,q \geq 0$, this is equivalent to
\begin{equation*}\begin{aligned}
    p = q \, \lor \, &p \cdot q = 0 \, \lor \Bigg(\exists a \in (0,1), \forall \xinX, \\
    &q(x) = b(a) \, p(x) \lor q(x) = \frac{1}{b(a)} \, p(x) \Bigg),
\end{aligned}\end{equation*}
where we defined $b(a) \coloneqq \parenth{\frac{1-a}{1+a}}^2 \in (0,1)$ for $a \in (0,1)$. Note that $b(a):(0,1) \to (0,1)$ is a bijection. Thus, the above is equivalent to
\begin{equation*}\begin{aligned}
    p = q \, \lor \, &p \cdot q = 0 \, \lor \Bigg(\exists b \in (0,1), \forall \xinX, \\
    &q(x) = b \, p(x) \lor q(x) = \frac{1}{b} \, p(x) \Bigg),
\end{aligned}\end{equation*}
as needed. Thus, with $p \coloneqq (\braket{e_x}{\rho \, e_x})_{\xinX}$ and $q \coloneqq (\braket{e_x}{\sigma \, e_x})_{\xinX}$, the lemma holds.
\end{proof}

\section{Random Density Operators}

\subsection{Random Quantum-Classical States}\label{RandomQuantumClassicalStates}

Our goal in this section is to randomly sample quantum-classical states $\rho \in \mathcal{D}(\mathcal{H}^A \otimes \mathcal{H}^B)$ of the form
\begin{equation}\label{RandomQuantumClassical}
    \rho = \sumk \alpha_k \rho_k \otimes \ketbra{f_k}{f_k},
\end{equation}
where the $\rho_k \in \mathcal{D}(\mathcal{H}^A)$, the $\alpha_k \geq 0$ satisfy $\sumk \alpha_k = 1$, and the $\ket{f_k}$ form an orthonormal basis for $\mathcal{H}^B$. To begin, we randomly sample the $\alpha_k$ from the $(d_B-1)$--dimensional simplex spanned by the standard basis vectors $(1,0,\ldots,0), \ldots, (0,\ldots,0,1) \in \mathbb{R}^{d_B}$. Then, we let $\set{\ket{e_j}}_j$ be a standard basis for $\mathcal{H}^A$ and $\set{\ket{f_k}}_k$ be a standard basis for $\mathcal{H}^B$. Next, for each $k$, we randomly choose the eigenvalues of $\rho_k$ from the $(d_A-1)$--dimensional simplex spanned by the standard basis vectors $(1,0,\ldots,0), \ldots, (0,\ldots,0,1) \in \mathbb{R}^{d_A}$, and we place those random eigenvalues on the main diagonal of a matrix $D_k \in \mathcal{L}(\mathcal{H}^A)$ (written with respect to the standard basis $\set{\ket{e_j}}_j$). To randomize the eigenvectors of the $\rho_k$, we pick random unitaries $U_k$ according to a Haar-uniform distribution, and we conjugate the $D_k$ by the $U_k$ to obtain the $\rho_k$. That is, we set $\rho_k \coloneqq U_k D_k U_k^{\dagger}$ for each $k$. Applying Eq.~(\ref{RandomQuantumClassical}) then yields the desired random quantum-classical states as used in Fig.~\ref{fig:RandomQuantumClassicalStates}. Mathematica code to implement the above procedure can be found on \github{Github}.

\subsection{Random Classical States with Fixed Angular Distance}\label{RandomClassicalStatesFixedAngularDistance}

Our goal in this section is to randomly sample classical (i.e.\ commuting) states $\rho, \sigma \in \mathcal{D}(\mathcal{H})$ with a prescribed angular distance $\A{\rho}{\sigma} = A$ for any $A \gtrsim 0$. To do this, consider any points $r, s \in \mathbb{R}^d$ on the unit $(d-1)$--sphere such that $r, s \geq 0$, where $d \coloneqq \dim \mathcal{H}$. Now, let $\rho$ and $\sigma$ be commuting states with eigenvalues $\rho_j \coloneqq r_j^2$ and $\sigma_j \coloneqq s_j^2$ respectively, where $j = 1, \ldots, d$. Then
\begin{equation*}\begin{aligned}
    \A{\rho}{\sigma} &= \arccos \norm{\srss}_1 \\
    &= \arccos \sum_{j=1}^d \sqrt{\rho_j \sigma_j} \\
    &= \arccos \sum_{j=1}^d r_j s_j \\
    &= \arccos r \cdot s,
\end{aligned}\end{equation*}
where the last line is just the angular distance between $r$ and $s$. Thus, by the above construction, it remains to randomly sample such points $r, s$. To do this, we randomly pick $\tilde{r}, \tilde{s}$ on the unit $(d-1)$--sphere. We then let $r = \operatorname{abs}(\tilde{r})$, where $\operatorname{abs}(\cdot)$ denotes component-wise absolute value. Next, we rotate $r$ towards $\tilde{s}$ by $A$ radians to obtain $s$, where $A \gtrsim 0$ is fixed. For $A$ close to 0, this procedure generates points $r, s$ in the positive hyperoctant of the unit $(d-1)$--sphere with high probability; $s$ is very unlikely to lie outside the positive hyperoctant for $A$ close enough to 0. Thus, we simply reject the few cases where $s$ happens to lie outside the positive hyperoctant. As explained above, this yields random classical states $\rho, \sigma$ at a fixed angular distance $A$, as used in Fig.~\ref{fig:FixedAngularDistance}. Mathematica code to implement the above procedure can be found on \github{Github}.

\section{Perturbation argument}
\label{app:perturb}

Consider any arbitrary (i.e.~possibly noninvertible) $\rho,\sigma$ saturating the upper Fuchs--van de Graaf inequality.
We continue our analysis onwards from Eq.~\eqref{eq:appFvdG}, attempting to apply a similar argument as in the proof of Theorem~\ref{QuantumFvdGTheorem}. Take any
$\Lambda_T \in \mathcal{T}(\rho, \sigma)$ (note that this step is essentially why we avoided perturbing the trace-distance term when writing Eq.~\eqref{eq:appFvdG} --- if we were to instead take
$\Lambda_T \in \mathcal{T}(\rho_\delta, \sigma_\delta)$, then $\Lambda_T$ would implicitly depend on $\delta$, which poses some challenges in our subsequent analysis). Now note that if we were to perform this measurement on the states $(\rho,\sigma)$, the trace distance between the resulting distributions is ``close'' to the trace distance between the distributions that would be obtained by performing this measurement on $(\rho_\delta,\sigma_\delta)$ instead:
\begin{align*}
&\left|\Tc{\tr{\Lambda_T \rho}}{\tr{\Lambda_T \sigma}} -  \Tc{\tr{\Lambda_T \rho_\delta}}{\tr{\Lambda_T \sigma_\delta}}\right| \\
\leq& \left|\Tc{\tr{\Lambda_T \rho}}{\tr{\Lambda_T \sigma}} -  \Tc{\tr{\Lambda_T \rho_\delta}}{\tr{\Lambda_T \sigma}}\right| \\
&+\left|\Tc{\tr{\Lambda_T \rho_\delta}}{\tr{\Lambda_T \sigma}} -  \Tc{\tr{\Lambda_T \rho_\delta}}{\tr{\Lambda_T \sigma_\delta}}\right| \\
\leq& \Tc{\tr{\Lambda_T \rho}}{\tr{\Lambda_T \rho_\delta}} + \Tc{\tr{\Lambda_T \sigma}}{\tr{\Lambda_T \sigma_\delta}} \\
\leq& \T{\rho}{\rho_\delta} + \T{\sigma}{\sigma_\delta} \eqqcolon f(\delta),
\end{align*}
where the second inequality holds due to the reverse triangle inequality for trace distance, and the function $f(\delta)$ in the last line satisfies $\lim_{\delta\to 0^+} f(\delta) = 0$.

With this, we can obtain the following chain of inequalities by following the same arguments as in the derivation of Eq.~\eqref{QuantumFvdGInequality2}:
\begin{equation*}
\begin{aligned}
    \T{\rho}{\sigma} - f(\delta)
    &= \Tc{\tr{\Lambda_T \rho}}{\tr{\Lambda_T \sigma}} - f(\delta) \\
    &\leq \Tc{\tr{\Lambda_T \rho_\delta}}{\tr{\Lambda_T \sigma_\delta}}  \\
    &\leq \sqrt{1-\Fc{\tr{\Lambda_T \rho_\delta}}{\tr{\Lambda_T \sigma_\delta}}^2}  \\
    &\leq\sqrt{1-\F{\rho_\delta}{\sigma_\delta}^2}\\
    &\leq \T{\rho}{\sigma} + g(\delta) ,
\end{aligned}
\end{equation*}
where the first inequality follows since $\T{\rho_\delta}{\sigma_\delta} \leq \T{\rho}{\sigma}$, and in the last line the function $g(\delta)$ represents the $\negl(\delta)$ bound in Eq.~\eqref{eq:appFvdG}. Observe that the first and last expressions in the above chain of inequalities differ by only $f(\delta) + g(\delta)$. This implies that for each individual inequality in the chain, the difference between the two sides of each inequality is also at most $f(\delta) + g(\delta)$, which is a negligible function $\negl(\delta)$. From this fact, and the continuity of the function $\sqrt{1-x^2}$, we conclude that the measurement $\Lambda_T$ necessarily satisfies
\begin{equation*}
    \left|\Fc{\tr{\Lambda_T \rho_\delta}}{\tr{\Lambda_T \sigma_\delta}} - \F{\rho_\delta}{\sigma_\delta}\right| \leq \negl(\delta),
\end{equation*}
i.e.~it ``approximately preserves'' the fidelity between $\rho_\delta,\sigma_\delta$.

This suggests that it may be useful to characterize the set of measurements that ``approximately preserve'' fidelity in the above sense. To this end, we prove the following lemma, which looks roughly similar in some ways to Lemma~\ref{VariationalFidelityLemma} (note, however, that we have only proven one direction of the implications in this lemma, i.e.~this result might not be a bidirectional implication).
\begin{lemma}
Consider any states $\rho, \sigma \in \mathcal{D}(\mathcal{H})$ and define $\rho_\delta, \sigma_\delta$ as in Eq.~\eqref{eq:perturb}, in which case $\rho_\delta, \sigma_\delta$ are invertible for all $\delta>0$, and we can define the operator
\begin{equation}\label{eq:defMdelta}
    M_\delta \coloneqq \rho_\delta^{-1} \# \sigma_\delta.
\end{equation}
Suppose that $\Lambda \in \mathcal{M}$ is a rank-1 projective measurement with projectors $\set{\ketbra{e_x}{e_x}}_{\xinX}$ such that for all $\delta>0$,
\begin{equation}\label{eq:approxF}
\left|\Fc{\tr{\Lambda \rho_\delta}}{\tr{\Lambda \sigma_\delta}} - \F{\rho_\delta}{\sigma_\delta}\right| 
\leq \negl(\delta).
\end{equation}
Then for any $\xinX$ such that $\ket{e_x} \notin \ker \rho$ and any $\delta>0$, there exists $\mu_{x \delta} \in \mathbb{C}$, which together satisfy
\begin{align}\label{eq:limitcond}
    \lim_{\delta\to 0^+} \sqrt{\rho_\delta} \left(M_\delta - \abs{\mu_{x \delta}} \mathbbm{1} \right) \ket{e_x} = 0.
\end{align}
\end{lemma}
\begin{proof}
To make some steps easier to follow, we shall start by writing the condition~\eqref{eq:approxF} in the form
\begin{equation}
\label{eq:approxFeps}
\left|\Fc{\tr{\Lambda \rho_\delta}}{\tr{\Lambda \sigma_\delta}} - \F{\rho_\delta}{\sigma_\delta}\right| 
\leq \eps_\delta,
\end{equation}
and only set $\eps_\delta = \negl(\delta)$ near the end of the argument. The idea of the proof is to follow essentially the same steps as in the proof in Appendix~\ref{ProofVariationalFidelityLemma}, except with ``approximate equalities'' instead of inequalities. We begin by letting $U_\delta$ be the operator (induced by polar decomposition) such that
\[
\sqrt{\sqrt{\rho_\delta} \, \sigma_\delta \sqrt{\rho_\delta}} = \sqrt{\rho_\delta}\sqrt{\sigma_\delta} U_\delta.
\]
For brevity, we introduce the following operators (note that this part of the construction also essentially works for a general measurement with POVM operators $\set{E_x}_{\xinX}$; one just needs to use $\sqrt{E_x}$ in place of $\ketbra{e_x}{e_x}$):
\[
A_{x \delta} \coloneqq \ketbra{e_x}{e_x} \sqrt{\sigma_\delta} U_\delta, \quad B_{x \delta} \coloneqq \ketbra{e_x}{e_x} \sqrt{\rho_\delta} .
\]
Note that since we are only considering the $\delta>0$ regime, $\rho_\delta,\sigma_\delta$ are invertible and hence the operators $A_{x \delta},B_{x \delta}$ are always nonzero.

Now, following the same calculations as in the proof in Appendix~\ref{ProofVariationalFidelityLemma} gives
\begin{equation*}\begin{aligned}
    \F{\rho_\delta}{\sigma_\delta}
    &= \sum_{\xinX} \tr{\sqrt{\rho_\delta} \, \ketbra{e_x}{e_x} \sqrt{\sigma_\delta} \, U_\delta} \\
    &= \sum_{\xinX} \tr{B_{x \delta}^\dagger A_{x \delta}} \\
    & \leq \sum_{\xinX} \abs{\tr{B_{x \delta}^\dagger A_{x \delta}}} \\
    & \leq \sum_{\xinX} \norm{A_{x \delta}}_2 \norm{B_{x \delta}}_2 \\
    &= \sum_{\xinX} \sqrt{\tr{\ketbra{e_x}{e_x} \, \rho_\delta}} \sqrt{\tr{\ketbra{e_x}{e_x} \, \sigma_\delta}} \\
    &= \Fc{\tr{\Lambda \rho_\delta}}{\tr{\Lambda \sigma_\delta}},
\end{aligned}\end{equation*}
where the second inequality is the step where Cauchy-Schwarz was applied. 
Combining this with the condition~\eqref{eq:approxFeps}, we see that both the inequalities in the chain must be ``tight up to $\eps_\delta$'', i.e.~we have
\begin{equation}\label{eq:approxbnd1}
    \abs{\sum_{\xinX} \abs{\tr{B_{x \delta}^\dagger A_{x \delta}}} - \sum_{\xinX} \tr{B_{x \delta}^\dagger A_{x \delta}} } \leq \eps_\delta
\end{equation}
and
\begin{equation}\label{eq:approxbnd2}
    \abs{\sum_{\xinX} \norm{A_{x \delta}}_2 \norm{B_{x \delta}}_2 - \sum_{\xinX} \tr{B_{x \delta}^\dagger A_{x \delta}} } \leq \eps_\delta.
\end{equation}

For the bound~\eqref{eq:approxbnd1}, if we write $z_x \coloneqq \tr{B_{x \delta}^\dagger A_{x \delta}}$, then we have some values $z_x \in \mathbb{C}$ such that $\abs{z_x}\leq \norm{A_{x \delta}}_2 \norm{B_{x \delta}}_2 \leq 1$, and their sum $\sum_{\xinX} z_x$ is real-valued (because it is equal to $\F{\rho_\delta}{\sigma_\delta}$) and within $\eps_\delta$ of the sum of their absolute values. Viewing these complex numbers $z_x$ as vectors in the complex plane, a geometric argument then shows that we must have 
$\abs{\abs{z_x} - z_x} \leq \sqrt{2\eps_\delta}$ for all $\xinX$, i.e.
\begin{equation}\label{eq:approxbndx1}
    \forall \xinX, \abs{\abs{\tr{B_{x \delta}^\dagger A_{x \delta}}} - \tr{B_{x \delta}^\dagger A_{x \delta}}} \leq \sqrt{2\eps_\delta}.
\end{equation}
As for the bound~\eqref{eq:approxbnd2}, note that the terms in the summations satisfy $\abs{\tr{B_{x \delta}^\dagger A_{x \delta}}} \leq \norm{A_{x \delta}}_2 \norm{B_{x \delta}}_2$ (this is just Cauchy-Schwarz) for all $\xinX$, from which we can deduce that
\begin{equation*}
    \forall \xinX, \abs{\, \norm{A_{x \delta}}_2 \norm{B_{x \delta}}_2 - \abs{\tr{B_{x \delta}^\dagger A_{x \delta}}} } \leq \eps_\delta.
\end{equation*}
The above bound suggests that $\norm{A_{x \delta}}_2^2 \norm{B_{x \delta}}_2^2$ and $\abs{\tr{B_{x \delta}^\dagger A_{x \delta}}}^2$ should be close as well (for all $\xinX$). To formalize this, we note that $\norm{A_{x \delta}}_2 \norm{B_{x \delta}}_2$ and $\abs{\tr{B_{x \delta}^\dagger A_{x \delta}}}$ are both upper bounded by $1$, and thus
\begin{align*}
    & \norm{A_{x \delta}}_2^2 \norm{B_{x \delta}}_2^2 - \abs{\tr{B_{x \delta}^\dagger A_{x \delta}}}^2 \\
    =& \left(\norm{A_{x \delta}}_2 \norm{B_{x \delta}}_2 + \abs{\tr{B_{x \delta}^\dagger A_{x \delta}}}\right)\cdot\\
    &\quad \left(\norm{A_{x \delta}}_2 \norm{B_{x \delta}}_2 - \abs{\tr{B_{x \delta}^\dagger A_{x \delta}}}\right)\\
    \leq& 2\left(\norm{A_{x \delta}}_2 \norm{B_{x \delta}}_2 - \abs{\tr{B_{x \delta}^\dagger A_{x \delta}}}\right).
\end{align*}
Combined with the preceding bound, this gives
\begin{equation}\label{eq:approxbndx2}
    \forall \xinX, \sqrt{\norm{A_{x \delta}}_2^2 \norm{B_{x \delta}}_2^2 - \abs{\tr{B_{x \delta}^\dagger A_{x \delta}}}^2} \leq \sqrt{2\eps_\delta}.
\end{equation}

The reason for expressing the bound in the above form is so we can now make use of the following {equality} that appears in the derivation of Cauchy-Schwarz (which can be verified by expanding the left-hand-side):
\begin{equation*}
\norm{A_{x \delta} - \frac{\tr{B_{x \delta}^\dagger A_{x \delta}}}{\norm{B_{x \delta}}_2^2} B_{x \delta} }_2^2 = \norm{A_{x \delta}}_2^2 - \frac{ \abs{\tr{B_{x \delta}^\dagger A_{x \delta}}}^2 }{\norm{B_{x \delta}}_2^2}.
\end{equation*}
Let us now define 
\begin{equation}\label{eq:defmu}
\mu_{x \delta} \coloneqq \frac{\tr{B_{x \delta}^\dagger A_{x \delta}}}{\norm{B_{x \delta}}_2^2},
\end{equation}
so that the left-hand-side of the preceding expression is just $\norm{A_{x \delta} - {\mu_{x \delta}} B_{x \delta} }_2^2$. 
With this, we can write
\begin{align*}
    \forall \xinX, \norm{A_{x \delta} - {\mu_{x \delta}} B_{x \delta} }_2 
    &= \sqrt{\norm{A_{x \delta}}_2^2 - \frac{ \abs{\tr{B_{x \delta}^\dagger A_{x \delta}}}^2 }{\norm{B_{x \delta}}_2^2}} \\
    &\leq \frac{\sqrt{2\eps_\delta}}{\norm{B_{x \delta}}_2},
\end{align*}
where in the last line we applied the bound~\eqref{eq:approxbnd2}.
With this, 
we have for all $\xinX$:
\begin{align*}
&\norm{A_{x \delta} - \abs{\mu_{x \delta}} B_{x \delta} }_2 \\ 
\leq& \norm{A_{x \delta} - {\mu_{x \delta}} B_{x \delta} }_2 + 
\abs{{\mu_{x \delta}} - \abs{\mu_{x \delta}}\,} \norm{B_{x \delta} }_2 
\\
=& \norm{A_{x \delta} - {\mu_{x \delta}} B_{x \delta} }_2 + \frac{\abs{\tr{B_{x \delta}^\dagger A_{x \delta}} - \abs{\tr{B_{x \delta}^\dagger A_{x \delta}}} }}{\norm{B_{x \delta} }_2} \\
\leq& \frac{\sqrt{2\eps_\delta}}{\norm{B_{x \delta}}_2} + \frac{\sqrt{2\eps_\delta}}{\norm{B_{x \delta}}_2} = \frac{2\sqrt{2\eps_\delta}}{\norm{B_{x \delta}}_2},
\end{align*}
where in the last line we applied the bound~\eqref{eq:approxbndx1}.

The above result is essentially the main bound that yields the desired claim. (Note that in the case where $\rho,\sigma$ are both invertible and we set both $\delta$ and $\eps_\delta$ to $0$, the above bound reduces to $A_{x \delta} = \abs{\mu_{x \delta}} B_{x \delta}$, which is basically equivalent to the $\ketbra{e_x}{e_x} \sqrt{\rho} = c_x \ketbra{e_x}{e_x}  \sqrt{\sigma} \, U$ condition in the Appendix~\ref{ProofVariationalFidelityLemma} proof.) To finish up, we note that by the definition of $B_{x \delta}$ we have $\norm{B_{x \delta}}_2 = \bra{e_x} \rho_\delta \ket{e_x}$, and hence (because $\lim_{\delta \to 0^+} \rho_\delta$ exists and equals $\rho$):
\[
\lim_{\delta \to 0^+} \norm{B_{x \delta}}_2 = \lim_{\delta \to 0^+} \bra{e_x} \rho_\delta \ket{e_x} = \bra{e_x} \rho \ket{e_x}.
\]
This means that for any $\xinX$ such that $\ket{e_x} \notin \ker \rho$, the value $L_x \coloneqq \lim_{\delta \to 0^+} \norm{B_{x \delta}}_2$ exists and is strictly positive.
Hence for such $x$, we know that for all sufficiently small $\delta$, we would have $\norm{B_{x \delta}}_2 \geq L_x/2$ and thus also
\begin{equation}\label{eq:convbnd}
\norm{A_{x \delta} - \abs{\mu_{x \delta}} B_{x \delta} }_2 
\leq \frac{4\sqrt{2\eps_\delta}}{L_x}.
\end{equation}
With this, we finally substitute $\eps_\delta = \negl(\delta)$ to conclude that for such $x$, we have
\begin{align*}
\lim_{\delta\to 0^+} \norm{A_{x \delta} - \abs{\mu_{x \delta}} B_{x \delta} }_2 = 0,
\end{align*}
and thus 
\begin{align}\label{eq:limitAB}
\lim_{\delta\to 0^+} (A_{x \delta} - \abs{\mu_{x \delta}} B_{x \delta}) = 0.
\end{align}
(It does not matter which operator norm is considered in the above convergence statement, because all norms on a finite-dimensional vector space yield the same topology.)
Note that for all $\delta>0$, by substituting the definitions of $A_{x \delta}$, $B_{x \delta}$ and $U_\delta$ (and also using the fact that $\rho_\delta$ is invertible), we get
\begin{align*}
A_{x \delta} - \abs{\mu_{x \delta}} B_{x \delta} &=  \ketbra{e_x}{e_x} (\sqrt{\sigma_\delta} U_\delta - \abs{\mu_{x \delta}} \sqrt{\rho_\delta}) \\
&=  \ketbra{e_x}{e_x} (M_\delta \sqrt{\rho_\delta}- \abs{\mu_{x \delta}} \sqrt{\rho_\delta}).
\end{align*}
Substituting this into \eqref{eq:limitAB}, then left-multiplying by $\bra{e_x}$ and taking the adjoint, we get the desired result.
\end{proof}

Note that if it can be shown that $\lim_{\delta\to 0^+} M_\delta$ and $\lim_{\delta\to 0^+} \mu_{x \delta}$ exist (let us denote their limiting values as $\lim_{\delta\to 0^+} M_\delta \eqqcolon M_0$ and $\lim_{\delta\to 0^+} \mu_{x \delta} \eqqcolon \mu_{x 0}$), then the
condition~\eqref{eq:limitcond} reduces to a form of ``skewed'' eigenvalue condition:
\begin{equation}\label{eq:skeweig}
    \sqrt{\rho} M_0 \ket{e_x} = \sqrt{\rho} \abs{\mu_{x 0}} \ket{e_x}.
\end{equation}
(In fact, it may not be strictly necessary to show that both $\lim_{\delta\to 0^+} M_\delta$ and $\lim_{\delta\to 0^+} \mu_{x \delta}$ exist; from the condition~\eqref{eq:limitcond} we know that e.g.~if $\lim_{\delta\to 0^+} \sqrt{\rho_\delta} M_\delta \ket{e_x} $ exists then so does $\lim_{\delta\to 0^+} \sqrt{\rho_\delta} \abs{\mu_{x \delta}} \ket{e_x} $ and vice versa, although the factors of $\sqrt{\rho_\delta}$ and $\ket{e_x}$ make this not entirely straightforward to work with).  
Roughly speaking, the main challenge in trying to get the above line of reasoning to yield a result fully similar to Lemma~\ref{VariationalFidelityLemma} is that when $\rho$ is noninvertible, we cannot multiply both sides of Eq.~\eqref{eq:skeweig} by $\sqrt{\rho^{-1}}$ to remove the $\sqrt{\rho}$ prefactors and get a genuine eigenvalue equation. If instead we try working with one of the intermediate bounds in the proof, such as~\eqref{eq:convbnd}, and multiply by $\sqrt{\rho_\delta^{-1}}$ (which is well-defined for $\delta>0$), the issue is that the maximum eigenvalue of $\sqrt{\rho_\delta^{-1}}$ diverges as $\delta\to 0^+$, making it difficult to bound the norm of the resulting quantities.

Regarding the question of whether $\lim_{\delta\to 0^+} M_\delta$ exists, we note that in the special case where both $\rho$ and $\sigma$ are pure qubit states, we can without loss of generality write $\rho = \ket{0}$ and $\sigma = \alpha \ket{0} + \beta \ket{1}$ in some basis, in which case we can compute the following limit for $\alpha\neq 0$:
\begin{equation*}
    \lim_{\delta \to 0^+} M_\delta = 
    \left(  
        \begin{array}{cc}
            \abs{\alpha} & \frac{a b^*}{\abs{\alpha}} \\
            \frac{b a^*}{\abs{\alpha}} & \frac{\abs{\beta}^2}{\abs{\alpha}}+\sqrt{1 + \frac{\abs{\beta}^2}{\abs{\alpha}^2}} \\
        \end{array}
    \right).
\end{equation*}
(The $\alpha=0$ case corresponds to $\rho,\sigma$ being orthogonal, in which case it appears that $M_\delta$ diverges as $\delta\to 0^+$, but this case is not very relevant in our context since the Fuchs--van de Graaf inequalities are trivially saturated in this case.)
More generally, for $\dim \mathcal{H} = 3$, we were able to compute the operator $M_\delta$ analytically in Mathematica and found that as long as $\rho,\sigma$ are nonorthogonal pure states, the limit $\lim_{\delta\to 0^+} M_\delta$ indeed exists. However, we currently do not have a generalization of the argument to higher dimensions.

Finally, we remark that one possible direction for further investigation is that rather than focusing on proving an analogue of Lemma~\ref{VariationalFidelityLemma}, we could instead try to use some of the intermediate steps in the proof (such as the bounds~\eqref{eq:approxbndx1} and~\eqref{eq:approxbndx2}) to more directly analyze the set of states saturating the upper Fuchs--van de Graaf inequality. However, it currently does not seem clear whether this gives a useful result.

\section{Applications in keyrate calculations}
\label{app:keyrates}

Here, we briefly outline some potential applications of our results in computing QKD keyrates, focusing on a form of QKD referred to as device-independent (DI) QKD~\cite{PAB+09}.
Basically, it was observed in e.g.~\cite{Woo14} that one potential approach to compute such keyrates is to lower bound the fidelity between a particular pair of states. If it could be shown that these states saturate the upper Fuchs--van de Graaf inequality, then the fidelity can be written in terms of the trace distance~\footnote{To be more precise: one could always use the lower Fuchs--van de Graaf inequality to bound the fidelity, but in this context the resulting values are highly suboptimal, hence we are interested in whether the states could in fact saturate the upper bound.}, which has an operational interpretation in terms of guessing probability~\cite{NC10}. Various methods are known for bounding guessing probabilities in DIQKD~\cite{MPA11,PM13,NBS+18}, and hence this could serve as another potential approach for DIQKD keyrate computations. In fact, if the upper Fuchs--van de Graaf inequality were saturated in this context, then the guessing-probability bound derived in~\cite{MPA11} for a particular DIQKD protocol would yield an expression exactly matching the (tight) keyrate formula derived in~\cite{PAB+09} for that protocol, suggesting some plausibility in this approach.

Another aspect of DIQKD in which such a result could be useful would be protocols based on advantage distillation~\cite{TLR20}, which refers to the use of two-way communication for the information-reconciliation step~\cite{rennerthesis} of the protocol. Again, the analysis in~\cite{TLR20} is based on the fidelity between a particular pair of states (different from the above), and it was observed that if that pair of states saturates the upper Fuchs--van de Graaf inequality, then significantly better results could be obtained. Hence a more detailed characterization of the set of such states could be of use.

\bibliography{main}

\begin{thebibliography}{32}%
\makeatletter
\providecommand \@ifxundefined [1]{%
 \@ifx{#1\undefined}
}%
\providecommand \@ifnum [1]{%
 \ifnum #1\expandafter \@firstoftwo
 \else \expandafter \@secondoftwo
 \fi
}%
\providecommand \@ifx [1]{%
 \ifx #1\expandafter \@firstoftwo
 \else \expandafter \@secondoftwo
 \fi
}%
\providecommand \natexlab [1]{#1}%
\providecommand \enquote  [1]{``#1''}%
\providecommand \bibnamefont  [1]{#1}%
\providecommand \bibfnamefont [1]{#1}%
\providecommand \citenamefont [1]{#1}%
\providecommand \href@noop [0]{\@secondoftwo}%
\providecommand \href [0]{\begingroup \@sanitize@url \@href}%
\providecommand \@href[1]{\@@startlink{#1}\@@href}%
\providecommand \@@href[1]{\endgroup#1\@@endlink}%
\providecommand \@sanitize@url [0]{\catcode `\\12\catcode `\$12\catcode
  `\&12\catcode `\#12\catcode `\^12\catcode `\_12\catcode `\%12\relax}%
\providecommand \@@startlink[1]{}%
\providecommand \@@endlink[0]{}%
\providecommand \url  [0]{\begingroup\@sanitize@url \@url }%
\providecommand \@url [1]{\endgroup\@href {#1}{\urlprefix }}%
\providecommand \urlprefix  [0]{URL }%
\providecommand \Eprint [0]{\href }%
\providecommand \doibase [0]{https://doi.org/}%
\providecommand \selectlanguage [0]{\@gobble}%
\providecommand \bibinfo  [0]{\@secondoftwo}%
\providecommand \bibfield  [0]{\@secondoftwo}%
\providecommand \translation [1]{[#1]}%
\providecommand \BibitemOpen [0]{}%
\providecommand \bibitemStop [0]{}%
\providecommand \bibitemNoStop [0]{.\EOS\space}%
\providecommand \EOS [0]{\spacefactor3000\relax}%
\providecommand \BibitemShut  [1]{\csname bibitem#1\endcsname}%
\let\auto@bib@innerbib\@empty
\bibitem [{\citenamefont {Wilde}(2013)}]{MW13}%
  \BibitemOpen
  \bibfield  {author} {\bibinfo {author} {\bibfnamefont {M.}~\bibnamefont
  {Wilde}},\ }\href {https://doi.org/https://doi.org/10.1017/CBO9781139525343}
  {\emph {\bibinfo {title} {Quantum Information Theory}}}\ (\bibinfo
  {publisher} {Cambridge University Press},\ \bibinfo {year}
  {2013})\BibitemShut {NoStop}%
\bibitem [{Note1()}]{Note1}%
  \BibitemOpen
  \bibinfo {note} {In this work, we define entropies via the natural logarithm
  rather than the base-$2$ logarithm for ease of presentation in the
  proofs.}\BibitemShut {Stop}%
\bibitem [{\citenamefont {Cleve}\ and\ \citenamefont
  {DiVincenzo}(1996)}]{CD96}%
  \BibitemOpen
  \bibfield  {author} {\bibinfo {author} {\bibfnamefont {R.}~\bibnamefont
  {Cleve}}\ and\ \bibinfo {author} {\bibfnamefont {D.~P.}\ \bibnamefont
  {DiVincenzo}},\ }\bibfield  {title} {\bibinfo {title} {Schumacher's quantum
  data compression as a quantum computation},\ }\href
  {https://doi.org/10.1103/physreva.54.2636} {\bibfield  {journal} {\bibinfo
  {journal} {Physical Review A}\ }\textbf {\bibinfo {volume} {54}},\ \bibinfo
  {pages} {2636} (\bibinfo {year} {1996})}\BibitemShut {NoStop}%
\bibitem [{\citenamefont {Berta}\ \emph {et~al.}(2012)\citenamefont {Berta},
  \citenamefont {Fawzi},\ and\ \citenamefont {Wehner}}]{BFW12}%
  \BibitemOpen
  \bibfield  {author} {\bibinfo {author} {\bibfnamefont {M.}~\bibnamefont
  {Berta}}, \bibinfo {author} {\bibfnamefont {O.}~\bibnamefont {Fawzi}},\ and\
  \bibinfo {author} {\bibfnamefont {S.}~\bibnamefont {Wehner}},\ }\bibfield
  {title} {\bibinfo {title} {Quantum to classical randomness extractors},\ }in\
  \href@noop {} {\emph {\bibinfo {booktitle} {Advances in Cryptology -- CRYPTO
  2012}}},\ \bibinfo {editor} {edited by\ \bibinfo {editor} {\bibfnamefont
  {R.}~\bibnamefont {Safavi-Naini}}\ and\ \bibinfo {editor} {\bibfnamefont
  {R.}~\bibnamefont {Canetti}}}\ (\bibinfo  {publisher} {Springer Berlin
  Heidelberg},\ \bibinfo {address} {Berlin, Heidelberg},\ \bibinfo {year}
  {2012})\ pp.\ \bibinfo {pages} {776--793}\BibitemShut {NoStop}%
\bibitem [{\citenamefont {Nielsen}\ and\ \citenamefont {Chuang}(2010)}]{NC10}%
  \BibitemOpen
  \bibfield  {author} {\bibinfo {author} {\bibfnamefont {M.~A.}\ \bibnamefont
  {Nielsen}}\ and\ \bibinfo {author} {\bibfnamefont {I.~L.}\ \bibnamefont
  {Chuang}},\ }\href@noop {} {\emph {\bibinfo {title} {{Quantum Computation and
  Quantum Information}}}}\ (\bibinfo  {publisher} {Cambridge University Press,
  New York},\ \bibinfo {year} {2010})\BibitemShut {NoStop}%
\bibitem [{\citenamefont {Audenaert}(2007)}]{KA07}%
  \BibitemOpen
  \bibfield  {author} {\bibinfo {author} {\bibfnamefont {K.}~\bibnamefont
  {Audenaert}},\ }\bibfield  {title} {\bibinfo {title} {{A sharp continuity
  estimate for the von Neumann entropy}},\ }\bibfield  {journal} {\bibinfo
  {journal} {Journal of Physics A: Mathematical and Theoretical}\ }\href
  {https://doi.org/https://doi.org/10.1088/1751-8113/40/28/S18}
  {https://doi.org/10.1088/1751-8113/40/28/S18} (\bibinfo {year}
  {2007})\BibitemShut {NoStop}%
\bibitem [{\citenamefont {Winter}(2016)}]{Win16}%
  \BibitemOpen
  \bibfield  {author} {\bibinfo {author} {\bibfnamefont {A.}~\bibnamefont
  {Winter}},\ }\bibfield  {title} {\bibinfo {title} {{Tight Uniform Continuity
  Bounds for Quantum Entropies: Conditional Entropy, Relative Entropy Distance
  and Energy Constraints}},\ }\href {https://doi.org/10.1007/s00220-016-2609-8}
  {\bibfield  {journal} {\bibinfo  {journal} {Communications in Mathematical
  Physics}\ }\textbf {\bibinfo {volume} {347}},\ \bibinfo {pages} {291}
  (\bibinfo {year} {2016})}\BibitemShut {NoStop}%
\bibitem [{\citenamefont {Upadhyaya}\ \emph {et~al.}(2021)\citenamefont
  {Upadhyaya}, \citenamefont {van Himbeeck}, \citenamefont {Lin},\ and\
  \citenamefont {L\"utkenhaus}}]{Upadhyaya_2021}%
  \BibitemOpen
  \bibfield  {author} {\bibinfo {author} {\bibfnamefont {T.}~\bibnamefont
  {Upadhyaya}}, \bibinfo {author} {\bibfnamefont {T.}~\bibnamefont {van
  Himbeeck}}, \bibinfo {author} {\bibfnamefont {J.}~\bibnamefont {Lin}},\ and\
  \bibinfo {author} {\bibfnamefont {N.}~\bibnamefont {L\"utkenhaus}},\
  }\bibfield  {title} {\bibinfo {title} {Dimension reduction in quantum key
  distribution for continuous- and discrete-variable protocols},\ }\href
  {https://doi.org/10.1103/PRXQuantum.2.020325} {\bibfield  {journal} {\bibinfo
   {journal} {PRX Quantum}\ }\textbf {\bibinfo {volume} {2}},\ \bibinfo {pages}
  {020325} (\bibinfo {year} {2021})}\BibitemShut {NoStop}%
\bibitem [{\citenamefont {Kanitschar}\ \emph {et~al.}(2023)\citenamefont
  {Kanitschar}, \citenamefont {George}, \citenamefont {Lin}, \citenamefont
  {Upadhyaya},\ and\ \citenamefont {Lütkenhaus}}]{KGL+23}%
  \BibitemOpen
  \bibfield  {author} {\bibinfo {author} {\bibfnamefont {F.}~\bibnamefont
  {Kanitschar}}, \bibinfo {author} {\bibfnamefont {I.}~\bibnamefont {George}},
  \bibinfo {author} {\bibfnamefont {J.}~\bibnamefont {Lin}}, \bibinfo {author}
  {\bibfnamefont {T.}~\bibnamefont {Upadhyaya}},\ and\ \bibinfo {author}
  {\bibfnamefont {N.}~\bibnamefont {Lütkenhaus}},\ }\bibfield  {title}
  {\bibinfo {title} {Finite-size security for discrete-modulated
  continuous-variable quantum key distribution protocols},\ }\Eprint
  {https://arxiv.org/abs/2301.08686} {arXiv:2301.08686}  (\bibinfo {year}
  {2023})\BibitemShut {NoStop}%
\bibitem [{\citenamefont {Sekatski}\ \emph {et~al.}(2021)\citenamefont
  {Sekatski}, \citenamefont {Bancal}, \citenamefont {Valcarce}, \citenamefont
  {Tan}, \citenamefont {Renner},\ and\ \citenamefont {Sangouard}}]{SBV+21}%
  \BibitemOpen
  \bibfield  {author} {\bibinfo {author} {\bibfnamefont {P.}~\bibnamefont
  {Sekatski}}, \bibinfo {author} {\bibfnamefont {J.-D.}\ \bibnamefont
  {Bancal}}, \bibinfo {author} {\bibfnamefont {X.}~\bibnamefont {Valcarce}},
  \bibinfo {author} {\bibfnamefont {E.~Y.-Z.}\ \bibnamefont {Tan}}, \bibinfo
  {author} {\bibfnamefont {R.}~\bibnamefont {Renner}},\ and\ \bibinfo {author}
  {\bibfnamefont {N.}~\bibnamefont {Sangouard}},\ }\bibfield  {title} {\bibinfo
  {title} {{Device-independent quantum key distribution from generalized CHSH
  inequalities}},\ }\href {https://doi.org/10.22331/q-2021-04-26-444}
  {\bibfield  {journal} {\bibinfo  {journal} {Quantum}\ }\textbf {\bibinfo
  {volume} {5}},\ \bibinfo {pages} {444} (\bibinfo {year} {2021})}\BibitemShut
  {NoStop}%
\bibitem [{\citenamefont {Mirsky}(1975)}]{Mirsky75}%
  \BibitemOpen
  \bibfield  {author} {\bibinfo {author} {\bibfnamefont {L.}~\bibnamefont
  {Mirsky}},\ }\bibfield  {title} {\bibinfo {title} {{A trace inequality of
  John von Neumann}},\ }\bibfield  {journal} {\bibinfo  {journal} {Monatshefte
  für Mathematik}\ }\href {https://doi.org/https://doi.org/10.1007/BF01647331}
  {https://doi.org/10.1007/BF01647331} (\bibinfo {year} {1975})\BibitemShut
  {NoStop}%
\bibitem [{\citenamefont {{Fuchs}}\ and\ \citenamefont {{van de
  Graaf}}(1999)}]{FvdG99}%
  \BibitemOpen
  \bibfield  {author} {\bibinfo {author} {\bibfnamefont {C.~A.}\ \bibnamefont
  {{Fuchs}}}\ and\ \bibinfo {author} {\bibfnamefont {J.}~\bibnamefont {{van de
  Graaf}}},\ }\bibfield  {title} {\bibinfo {title} {{Cryptographic
  distinguishability measures for quantum-mechanical states}},\ }\href
  {https://doi.org/10.1109/18.761271} {\bibfield  {journal} {\bibinfo
  {journal} {IEEE Transactions on Information Theory}\ }\textbf {\bibinfo
  {volume} {45}},\ \bibinfo {pages} {1216} (\bibinfo {year}
  {1999})}\BibitemShut {NoStop}%
\bibitem [{\citenamefont {Ando}\ \emph {et~al.}(2004)\citenamefont {Ando},
  \citenamefont {Li},\ and\ \citenamefont {Mathias}}]{Ando04}%
  \BibitemOpen
  \bibfield  {author} {\bibinfo {author} {\bibfnamefont {T.}~\bibnamefont
  {Ando}}, \bibinfo {author} {\bibfnamefont {C.-K.}\ \bibnamefont {Li}},\ and\
  \bibinfo {author} {\bibfnamefont {R.}~\bibnamefont {Mathias}},\ }\bibfield
  {title} {\bibinfo {title} {Geometric means},\ }\bibfield  {journal} {\bibinfo
   {journal} {Linear Algebra and its Applications}\ }\href
  {https://doi.org/https://doi.org/10.1016/j.laa.2003.11.019}
  {https://doi.org/10.1016/j.laa.2003.11.019} (\bibinfo {year}
  {2004})\BibitemShut {NoStop}%
\bibitem [{\citenamefont {Toussaint}(1972)}]{Toussaint72}%
  \BibitemOpen
  \bibfield  {author} {\bibinfo {author} {\bibfnamefont {G.}~\bibnamefont
  {Toussaint}},\ }\bibfield  {title} {\bibinfo {title} {{Comments on ``The
  Divergence and Bhattacharyya Distance Measures in Signal Selection"}},\
  }\bibfield  {journal} {\bibinfo  {journal} {IEEE Transactions on
  Communications}\ }\href
  {https://doi.org/https://doi.org/10.1109/TCOM.1972.1091157}
  {https://doi.org/10.1109/TCOM.1972.1091157} (\bibinfo {year}
  {1972})\BibitemShut {NoStop}%
\bibitem [{\citenamefont {Helstrom}(1976)}]{Helstrom76}%
  \BibitemOpen
  \bibfield  {author} {\bibinfo {author} {\bibfnamefont {C.}~\bibnamefont
  {Helstrom}},\ }\href@noop {} {\emph {\bibinfo {title} {Quantum Detection and
  Estimation Theory}}}\ (\bibinfo  {publisher} {Academic Press},\ \bibinfo
  {year} {1976})\BibitemShut {NoStop}%
\bibitem [{\citenamefont {Fuchs}\ and\ \citenamefont {Caves}(1995)}]{FC95}%
  \BibitemOpen
  \bibfield  {author} {\bibinfo {author} {\bibfnamefont {C.~A.}\ \bibnamefont
  {Fuchs}}\ and\ \bibinfo {author} {\bibfnamefont {C.~M.}\ \bibnamefont
  {Caves}},\ }\bibfield  {title} {\bibinfo {title} {{Mathematical techniques
  for quantum communication theory}},\ }\href
  {https://doi.org/10.1007/BF02228997} {\bibfield  {journal} {\bibinfo
  {journal} {Open Systems {\&} Information Dynamics}\ }\textbf {\bibinfo
  {volume} {3}},\ \bibinfo {pages} {345} (\bibinfo {year} {1995})}\BibitemShut
  {NoStop}%
\bibitem [{\citenamefont {Bhatia}(2006)}]{Bha06}%
  \BibitemOpen
  \bibfield  {author} {\bibinfo {author} {\bibfnamefont {R.}~\bibnamefont
  {Bhatia}},\ }\href@noop {} {\emph {\bibinfo {title} {{Positive definite
  matrices}}}},\ Princeton Series in Applied Mathematics\ (\bibinfo
  {publisher} {Princeton University Press},\ \bibinfo {address} {Princeton,
  NJ},\ \bibinfo {year} {2006})\BibitemShut {NoStop}%
\bibitem [{\citenamefont {Hanson}\ and\ \citenamefont
  {Datta}(2019)}]{arx_HD19}%
  \BibitemOpen
  \bibfield  {author} {\bibinfo {author} {\bibfnamefont {E.~P.}\ \bibnamefont
  {Hanson}}\ and\ \bibinfo {author} {\bibfnamefont {N.}~\bibnamefont {Datta}},\
  }\bibfield  {title} {\bibinfo {title} {{Universal proofs of entropic
  continuity bounds via majorization flow}},\ }\href
  {https://arxiv.org/abs/1909.06981v3} {\bibfield  {journal} {\bibinfo
  {journal} {arXiv:1909.06981v3 [quant-ph]}\ } (\bibinfo {year}
  {2019})}\BibitemShut {NoStop}%
\bibitem [{\citenamefont {Jabbour}\ and\ \citenamefont
  {Datta}(2020)}]{arx_JD20}%
  \BibitemOpen
  \bibfield  {author} {\bibinfo {author} {\bibfnamefont {M.~G.}\ \bibnamefont
  {Jabbour}}\ and\ \bibinfo {author} {\bibfnamefont {N.}~\bibnamefont
  {Datta}},\ }\bibfield  {title} {\bibinfo {title} {{A tight uniform continuity
  bound for the Arimoto-R\'{e}nyi conditional entropy and its extension to
  classical-quantum states}},\ }\href {https://arxiv.org/abs/2007.05049v3}
  {\bibfield  {journal} {\bibinfo  {journal} {arXiv:2007.05049v3 [cs.IT]}\ }
  (\bibinfo {year} {2020})}\BibitemShut {NoStop}%
\bibitem [{\citenamefont {Marwah}\ and\ \citenamefont
  {Dupuis}(2022)}]{arx_MD22}%
  \BibitemOpen
  \bibfield  {author} {\bibinfo {author} {\bibfnamefont {A.}~\bibnamefont
  {Marwah}}\ and\ \bibinfo {author} {\bibfnamefont {F.}~\bibnamefont
  {Dupuis}},\ }\bibfield  {title} {\bibinfo {title} {{Uniform continuity bound
  for sandwiched R\'{e}nyi conditional entropy}},\ }\href
  {https://arxiv.org/abs/2201.05534v2} {\bibfield  {journal} {\bibinfo
  {journal} {arXiv:2201.05534v2 [quant-ph]}\ } (\bibinfo {year}
  {2022})}\BibitemShut {NoStop}%
\bibitem [{\citenamefont {Bluhm}\ \emph {et~al.}(2022)\citenamefont {Bluhm},
  \citenamefont {Capel}, \citenamefont {Gondolf},\ and\ \citenamefont
  {P\'{e}rez-Hern\'{a}ndez}}]{arx_BCG+22}%
  \BibitemOpen
  \bibfield  {author} {\bibinfo {author} {\bibfnamefont {A.}~\bibnamefont
  {Bluhm}}, \bibinfo {author} {\bibfnamefont {{\'{A}}.}~\bibnamefont {Capel}},
  \bibinfo {author} {\bibfnamefont {P.}~\bibnamefont {Gondolf}},\ and\ \bibinfo
  {author} {\bibfnamefont {A.}~\bibnamefont {P\'{e}rez-Hern\'{a}ndez}},\
  }\bibfield  {title} {\bibinfo {title} {{Continuity of quantum entropic
  quantities via almost convexity}},\ }\href
  {https://arxiv.org/abs/2208.00922v1} {\bibfield  {journal} {\bibinfo
  {journal} {arXiv:2208.00922v1 [quant-ph]}\ } (\bibinfo {year}
  {2022})}\BibitemShut {NoStop}%
\bibitem [{\citenamefont {Wilde}(2020)}]{Wil20}%
  \BibitemOpen
  \bibfield  {author} {\bibinfo {author} {\bibfnamefont {M.~M.}\ \bibnamefont
  {Wilde}},\ }\bibfield  {title} {\bibinfo {title} {{Optimal uniform continuity
  bound for conditional entropy of classical{\textendash}quantum states}},\
  }\bibfield  {journal} {\bibinfo  {journal} {Quantum Information Processing}\
  }\textbf {\bibinfo {volume} {19}},\ \href
  {https://doi.org/10.1007/s11128-019-2563-4} {10.1007/s11128-019-2563-4}
  (\bibinfo {year} {2020})\BibitemShut {NoStop}%
\bibitem [{Note2()}]{Note2}%
  \BibitemOpen
  \bibinfo {note} {To outline the key ideas, in~\cite {RFZ10} the following
  bound was derived (see Eq.~(21) of that work): $\protect \operatorname {H}
  \left ( A|B \right )_{\rho } \geq 1-\protect \operatorname {h} \left ( \left
  (1-\protect \operatorname {F} \left ( \rho ^{(0)}_{B}, \rho ^{(1)}_{B} \right
  )\right )/2 \right )$ for $\rho _{AB}$ of the form $\rho _{AB} = \DOTSB \sum@
  \slimits@ _{a\in \{0,1\}} (1/2) \left | a \right \rangle \protect \! \left
  \langle a \right | \otimes \rho ^{(a)}_{B}$, with equality holding when the
  two conditional states $\rho ^{(a)}_{B}$ are pure. Take any $\rho _{AB}$ such
  that this bound is a strict inequality, then observe that any purifications
  of the conditional states must satisfy $\protect \operatorname {F} \left (
  \rho ^{(0)}_{B}, \rho ^{(1)}_{B} \right ) \geq \protect \operatorname {F}
  \left ( \rho ^{(0)}_{BR}, \rho ^{(1)}_{BR} \right )$, and use the fact that
  the bound from~\cite {RFZ10} becomes an equality for pure $\rho
  ^{(a)}_{BR}$.}\BibitemShut {Stop}%
\bibitem [{\citenamefont {Pironio}\ \emph {et~al.}(2009)\citenamefont
  {Pironio}, \citenamefont {Ac\'in}, \citenamefont {Brunner}, \citenamefont
  {Gisin}, \citenamefont {Massar},\ and\ \citenamefont {Scarani}}]{PAB+09}%
  \BibitemOpen
  \bibfield  {author} {\bibinfo {author} {\bibfnamefont {S.}~\bibnamefont
  {Pironio}}, \bibinfo {author} {\bibfnamefont {A.}~\bibnamefont {Ac\'in}},
  \bibinfo {author} {\bibfnamefont {N.}~\bibnamefont {Brunner}}, \bibinfo
  {author} {\bibfnamefont {N.}~\bibnamefont {Gisin}}, \bibinfo {author}
  {\bibfnamefont {S.}~\bibnamefont {Massar}},\ and\ \bibinfo {author}
  {\bibfnamefont {V.}~\bibnamefont {Scarani}},\ }\bibfield  {title} {\bibinfo
  {title} {{Device-independent quantum key distribution secure against
  collective attacks}},\ }\href {https://doi.org/10.1088/1367-2630/11/4/045021}
  {\bibfield  {journal} {\bibinfo  {journal} {New Journal of Physics}\ }\textbf
  {\bibinfo {volume} {11}},\ \bibinfo {pages} {045021} (\bibinfo {year}
  {2009})}\BibitemShut {NoStop}%
\bibitem [{\citenamefont {Woodhead}(2014)}]{Woo14}%
  \BibitemOpen
  \bibfield  {author} {\bibinfo {author} {\bibfnamefont {E.}~\bibnamefont
  {Woodhead}},\ }\bibfield  {title} {\bibinfo {title} {{Tight asymptotic key
  rate for the Bennett-Brassard 1984 protocol with local randomization and
  device imprecisions}},\ }\href {https://doi.org/10.1103/PhysRevA.90.022306}
  {\bibfield  {journal} {\bibinfo  {journal} {Physical Review A}\ }\textbf
  {\bibinfo {volume} {90}},\ \bibinfo {pages} {022306} (\bibinfo {year}
  {2014})}\BibitemShut {NoStop}%
\bibitem [{Note3()}]{Note3}%
  \BibitemOpen
  \bibinfo {note} {To be more precise: one could always use the lower
  Fuchs--van de Graaf inequality to bound the fidelity, but in this context the
  resulting values are highly suboptimal, hence we are interested in whether
  the states could in fact saturate the upper bound.}\BibitemShut {Stop}%
\bibitem [{\citenamefont {Masanes}\ \emph {et~al.}(2011)\citenamefont
  {Masanes}, \citenamefont {Pironio},\ and\ \citenamefont
  {Ac{\'{\i}}n}}]{MPA11}%
  \BibitemOpen
  \bibfield  {author} {\bibinfo {author} {\bibfnamefont {L.}~\bibnamefont
  {Masanes}}, \bibinfo {author} {\bibfnamefont {S.}~\bibnamefont {Pironio}},\
  and\ \bibinfo {author} {\bibfnamefont {A.}~\bibnamefont {Ac{\'{\i}}n}},\
  }\bibfield  {title} {\bibinfo {title} {{Secure device-independent quantum key
  distribution with causally independent measurement devices}},\ }\bibfield
  {journal} {\bibinfo  {journal} {Nature Communications}\ }\textbf {\bibinfo
  {volume} {2}},\ \href {https://doi.org/10.1038/ncomms1244}
  {10.1038/ncomms1244} (\bibinfo {year} {2011})\BibitemShut {NoStop}%
\bibitem [{\citenamefont {Pironio}\ and\ \citenamefont {Massar}(2013)}]{PM13}%
  \BibitemOpen
  \bibfield  {author} {\bibinfo {author} {\bibfnamefont {S.}~\bibnamefont
  {Pironio}}\ and\ \bibinfo {author} {\bibfnamefont {S.}~\bibnamefont
  {Massar}},\ }\bibfield  {title} {\bibinfo {title} {{Security of practical
  private randomness generation}},\ }\href
  {https://doi.org/10.1103/PhysRevA.87.012336} {\bibfield  {journal} {\bibinfo
  {journal} {Physical Review A}\ }\textbf {\bibinfo {volume} {87}},\ \bibinfo
  {pages} {012336} (\bibinfo {year} {2013})}\BibitemShut {NoStop}%
\bibitem [{\citenamefont {Nieto-Silleras}\ \emph {et~al.}(2018)\citenamefont
  {Nieto-Silleras}, \citenamefont {Bamps}, \citenamefont {Silman},\ and\
  \citenamefont {Pironio}}]{NBS+18}%
  \BibitemOpen
  \bibfield  {author} {\bibinfo {author} {\bibfnamefont {O.}~\bibnamefont
  {Nieto-Silleras}}, \bibinfo {author} {\bibfnamefont {C.}~\bibnamefont
  {Bamps}}, \bibinfo {author} {\bibfnamefont {J.}~\bibnamefont {Silman}},\ and\
  \bibinfo {author} {\bibfnamefont {S.}~\bibnamefont {Pironio}},\ }\bibfield
  {title} {\bibinfo {title} {{Device-independent randomness generation from
  several Bell estimators}},\ }\href {https://doi.org/10.1088/1367-2630/aaaa06}
  {\bibfield  {journal} {\bibinfo  {journal} {New Journal of Physics}\ }\textbf
  {\bibinfo {volume} {20}},\ \bibinfo {pages} {023049} (\bibinfo {year}
  {2018})}\BibitemShut {NoStop}%
\bibitem [{\citenamefont {Tan}\ \emph {et~al.}(2020)\citenamefont {Tan},
  \citenamefont {Lim},\ and\ \citenamefont {Renner}}]{TLR20}%
  \BibitemOpen
  \bibfield  {author} {\bibinfo {author} {\bibfnamefont {E.~Y.-Z.}\
  \bibnamefont {Tan}}, \bibinfo {author} {\bibfnamefont {C.~C.-W.}\
  \bibnamefont {Lim}},\ and\ \bibinfo {author} {\bibfnamefont {R.}~\bibnamefont
  {Renner}},\ }\bibfield  {title} {\bibinfo {title} {{Advantage Distillation
  for Device-Independent Quantum Key Distribution}},\ }\href
  {https://doi.org/10.1103/PhysRevLett.124.020502} {\bibfield  {journal}
  {\bibinfo  {journal} {Physical Review Letters}\ }\textbf {\bibinfo {volume}
  {124}},\ \bibinfo {pages} {020502} (\bibinfo {year} {2020})}\BibitemShut
  {NoStop}%
\bibitem [{\citenamefont {Renner}(2005)}]{rennerthesis}%
  \BibitemOpen
  \bibfield  {author} {\bibinfo {author} {\bibfnamefont {R.}~\bibnamefont
  {Renner}},\ }\href {https://doi.org/10.3929/ethz-a-005115027} {\bibinfo
  {title} {{Security of Quantum Key Distribution}}} (\bibinfo {year}
  {2005})\BibitemShut {NoStop}%
\bibitem [{\citenamefont {Roga}\ \emph {et~al.}(2010)\citenamefont {Roga},
  \citenamefont {Fannes},\ and\ \citenamefont {\.{Z}yczkowski}}]{RFZ10}%
  \BibitemOpen
  \bibfield  {author} {\bibinfo {author} {\bibfnamefont {W.}~\bibnamefont
  {Roga}}, \bibinfo {author} {\bibfnamefont {M.}~\bibnamefont {Fannes}},\ and\
  \bibinfo {author} {\bibfnamefont {K.}~\bibnamefont {\.{Z}yczkowski}},\
  }\bibfield  {title} {\bibinfo {title} {{Universal Bounds for the Holevo
  Quantity, Coherent Information, and the Jensen-Shannon Divergence}},\ }\href
  {https://doi.org/10.1103/PhysRevLett.105.040505} {\bibfield  {journal}
  {\bibinfo  {journal} {Physical Review Letters}\ }\textbf {\bibinfo {volume}
  {105}},\ \bibinfo {pages} {040505} (\bibinfo {year} {2010})}\BibitemShut
  {NoStop}%
\end{thebibliography}%

\end{document}